\documentclass[conference]{IEEEtran}
\IEEEoverridecommandlockouts

\usepackage[utf8]{inputenc}
\usepackage{amssymb}
\usepackage{amsmath}
\usepackage{amsthm}
\usepackage{amsfonts}
\usepackage{hyperref}
\hypersetup{hidelinks,breaklinks}
\usepackage{stmaryrd}
\usepackage{microtype}
\usepackage{cleveref}
\usepackage[textsize=small]{todonotes}
\usepackage{thmtools,thm-restate}
\usepackage{mathtools}
\usepackage[normalem]{ulem}

\newtheorem{theorem}{Theorem}
\newtheorem{lemma}[theorem]{Lemma}
\newtheorem{corollary}[theorem]{Corollary}

\newtheorem{remark}[theorem]{Remark}

\newtheorem{proposition}[theorem]{Proposition}
\newtheorem{example}{Example}

\newcommand{\N}{\mathbb N}
\newcommand{\Q}{\mathbb Q}
\newcommand{\R}{\mathbb R}

\newcommand{\Z}{\mathbb Z}

\newcommand{\PowS}[2]{#1[[#2]]}

\newcommand{\poly}{\mbox{{\bf poly}}}

\newcommand{\set}[1]{\{#1\}}

\newcommand{\sem}[1]{\left\llbracket#1\right\rrbracket}

\newcommand{\ord}{\mathrm{ord}}

\newif\ifstartedinmathmode
\newcommand*{\st}{
  \relax\ifmmode\startedinmathmodetrue\else\startedinmathmodefalse\fi
  \ifstartedinmathmode{\;\cdot\;}\else{s.t.~}\fi%
}

\newcommand{\tuple}[1]{\left(#1\right)}

\renewcommand{\vec}[1]{\mathbf{#1}}

\newcommand{\coRP}{\mbox{{\bf coRP}}}

\newcommand{\Pt}{\mbox{{\bf P}}}
\newcommand{\FPt}{\mbox{{\bf FP}}}

\newcommand{\PP}{\mbox{{\bf PP}}}
\newcommand{\coNP}{\mbox{{\bf coNP}}}

\newcommand{\NP}{\mbox{{\bf NP}}}

\newcommand{\sharpP}{\mbox{{\#\bf P}}}

\newcommand{\pspace}{{\bf{PSPACE}}}

\crefname{section}{Sec.}{Sections}
\crefname{fact}{Fact}{Facts}
\crefname{claim}{Claim}{Claims}
\crefname{proposition}{Proposition}{Propositions}
\crefname{lemma}{Lemma}{Lemmas}
\crefname{corollary}{Corollary}{Corollaries}
\crefname{theorem}{Theorem}{Theorems}
\crefname{figure}{Figure}{Figures}
\crefname{equation}{Eq.}{Eqns.}
\crefname{appendix}{Appendix}{Appendices}

\definecolor{darkgreen}{rgb}{0, 0.8, 0}
\definecolor{darkred}{rgb}{0.8, 0, 0}

\newcommand{\CoeffSLP}{$\mathsf{CoeffSLP}$}
\newcommand{\EqSLP}{$\mathsf{EqSLP}$}
\newcommand{\DegSLP}{$\mathsf{DegSLP}$}
\newcommand{\PosSLP}{$\mathsf{PosSLP}$}
\newcommand{\CoeffAlg}{$\mathsf{CoeffAlg}$}
\newcommand{\EqAlg}{$\mathsf{EqAlg}$}
\newcommand{\FinAlg}{$\mathsf{FinAlg}$}

\begin{document}


\title{%
    Multiplicity Problems on Algebraic Series and Context-Free Grammars%
    \thanks{%
        Nikhil Balaji and Mahsa Shirmohammadi are supported by International Emerging Actions grant (IEA'22).
        Mahsa Shirmohammadi is supported by ANR grant VeSyAM (ANR-22-CE48-0005). 
        James Worrell is supported by EPSRC fellowship EP/X033813/1.
        Lorenzo Clemente is partially supported by the Polish NCN grant 2017/26/D/ST6/00201
        and by the European Research Council (ERC) project INFSYS (``Challenging Problems in Infinite-State Systems'', grant agreement No. 950398).%
    }%
}


\author{\IEEEauthorblockN{Nikhil Balaji\IEEEauthorrefmark{1},
Lorenzo Clemente\IEEEauthorrefmark{2},
Klara Nosan\IEEEauthorrefmark{3}, 
Mahsa Shirmohammadi\IEEEauthorrefmark{3} and
James Worrell\IEEEauthorrefmark{4}}
\IEEEauthorblockA{\IEEEauthorrefmark{1}IIT Delhi, India}
\IEEEauthorblockA{\IEEEauthorrefmark{2}University of Warsaw, Poland}
\IEEEauthorblockA{\IEEEauthorrefmark{3}Universit\'e Paris Cité, CNRS, IRIF, France}
\IEEEauthorblockA{\IEEEauthorrefmark{4}Department of Computer Science, University of Oxford, UK}}


\maketitle

\begin{abstract}
In this paper we obtain complexity bounds for computational problems
on algebraic power series over several commuting variables.
The power series are specified   by
systems of polynomial equations: a formalism closely related to  weighted context-free grammars.
We focus on
three problems---decide whether a given algebraic series
is identically zero, determine whether all but finitely many
coefficients are zero, and compute the coefficient of a specific
monomial.  We relate these questions to well-known computational
problems on arithmetic circuits and thereby show that all three
problems lie in the counting hierarchy.  Our main result improves the
best known complexity bound on deciding zeroness of an algebraic
series. This problem is known to lie in PSPACE by reduction
to the decision problem for the existential fragment of the theory of
real closed fields. Here we show that the problem lies in the counting
hierarchy by reduction to the problem of computing the degree of a
polynomial given by an arithmetic circuit.  As a corollary we obtain
new complexity bounds on multiplicity equivalence of context-free grammars restricted to a bounded language,
language inclusion of a non-deterministic finite automaton in an unambiguous context-free grammar,
and language inclusion of a non-deterministic context-free grammar in an unambiguous finite automaton.
\end{abstract}


\section{Introduction}
The subject of this paper is algebraic power series---formal 
series that satisfy a polynomial equation over the field of rational functions.  For
example, consider the generating function $C(x):=\sum_{n=0}^\infty C_n
x^n$ of the sequence $(C_n)_{n=0}^\infty$
of Catalan numbers. Recall that $C_n$ is the number of Dyck words of length $2n$.  The
series $C(x)$ satisfies the polynomial
equation  $1-C(x)+xC(x)^2 =0$ and hence is algebraic.
Algebraic power series generalise
rational series (which are the generating functions of linear
recurrence sequences) and are a subclass of D-finite power series (which are the
generating functions of holonomic sequences).  To illustrate the latter inclusion, note that the Catalan
numbers satisfy the holonomic recurrence $(n+2)C_{n+1} =
2(2n+1)C_n$. 

Algebraic power series 
have been an object of study in
formal language theory ever since
the seminal work of Chomsky and Sch\"utzenberger
(see~\cite[Chapter~IV]{SalomaaSoittola}
and~\cite[Chapters~II and~II]{KuichSalomaa} 
for the principal results and bibliographic references).
In this framework one specifies algebraic series via so-called \emph{proper systems of polynomial equations}, which can alternately be seen as weighted context-free grammars  or as a generalisation of arithmetic circuits that allows cycles.
Such equation systems can be defined both over
commuting and non-commuting variables.
In this work we focus on equation systems over 
several commuting variables with weights in $\Z$.

Two fundamental computational problems associated with an
algebraic series are to compute its coefficients and to
determine whether or not the series is identically zero.
These are generalisations of two well-studied problems on circuits, namely the problem
{\CoeffSLP} of determining a given coefficient 
of the polynomial represented by an arithmetic circuit, and the problem {\EqSLP}
of deciding zeroness of such a polynomial.  (The acronym SLP here stands for \emph{straight line program}, which we treat as synonymous with the term arithmetic circuit.  The problem {\EqSLP} is more commonly called \emph{arithmetic circuit identity testing} or \emph{polynomial identity testing}.)

We extend the problem {\CoeffSLP} from circuits to algebraic series, obtaining the
problem {\CoeffAlg}. The input to the latter
is a system of polynomial equations whose unique solution is a
multivariate 
power series $\sum_{\vec{v}\in \N^k} a_\vec{v} X^{\vec{v}}$ with integer coefficients, a vector $\vec{v}\in \N^k$, and a prime number $p$ (where $\vec{v}$, $p$, and all coefficients of the system of polynomial equations are given in binary).  The problem asks to determine $a_{\vec{v}} \bmod p$.  The reason to introduce the  modulus $p$ is because (as already in the circuit case) 
the bit-length of $a_{\vec{v}}$ is potentially exponential in the bit-length of~$\vec{v}$.

Our approach to {\CoeffAlg} involves a multivariate version of
Hensel's Lemma for computing zeros of systems of polynomials in the
ring of formal power series.  The advantage of this approach
is that, thanks to the quadratic convergence of the root approximation in Hensel's Lemma (meaning that the precision of the approximation doubles with each iteration),  one can compute the $n$-th
coefficient of a power series using a number of arithmetic operations that is polynomial in the bit length of~$n$.   We exploit this fact to reduce 
{\CoeffAlg} to {\CoeffSLP}.
The latter problem is known to be $\sharpP$-hard (see for example \cite{AllenderBurgisserKjeldgaard-PedersenMiltersen:JC:2009})
and from the proof of \cite[Theorem 4.1]{kayal-saha} it can be deduced to be in $\FPt^{\sharpP}$.
We thus have:
\begin{restatable}{theorem}{thmCoeffAlg} \label{thm:CoeffAlg}
    {\CoeffAlg} is equivalent under polynomial time reductions to {\CoeffSLP} and hence is $\sharpP$-hard and in ${\FPt}^{\sharpP}$.
\end{restatable}
\noindent
To the best of our knowledge, the theoretical complexity of \CoeffAlg~has not been studied before.

Our second main result concerns the complexity of the problem {\EqAlg}: \emph{given a proper system of polynomial equations, determine whether all coefficients of its power series solution are zero.}  
This problem was shown in~\cite[Chapter IV, Theorem 5.1]{SalomaaSoittola} to be polynomial-time reducible to the decision problem for the theory of real-closed fields and later in~\cite[Lemma 1]{FJKW} to be reducible to the 
decision problem  
for the \emph{existential} fragment of this theory, which we denote $\exists\mathbb{R}$.
Here we improve the complexity upper bound on {\EqAlg} by showing
that it is polynomial-time
reducible to the problem {\DegSLP}: \emph{given a polynomial $f$ represented by an algebraic circuit and an integer $d$, decide whether $\deg f\leq d$.}
The substance of this improvement is
twofold: first,
{\DegSLP} is polynomial-time reducible to
$\exists \mathbb{R}$\footnote{
The reduction follows from~\cite{AllenderBurgisserKjeldgaard-PedersenMiltersen:JC:2009}, where {\DegSLP} is shown to be polynomial-time reducible to the problem {\PosSLP} (decide whether a variable-free arithmetic circuit denotes a positive integer), which in turn belongs to $\exists \mathbb{R}$.
}; second, {\DegSLP} is known to lie in the counting hierarchy
whereas the best upper bound for $\exists \mathbb{R}$ is 
$\pspace$.  It is also worth noting that {\DegSLP} is not known to be \NP-hard, whereas 
$\exists\mathbb{R}$ is trivially so.  

\begin{restatable}{theorem}{thmAlgIT} \label{thm:AlgIT}
{\EqAlg} is polynomial-time reducible to {\DegSLP} and thereby lies in the counting class 
$\coRP^{\PP}$.
\end{restatable}

The proof of Theorem~\ref{thm:AlgIT} combines Hensel's Lemma (as in Theorem~\ref{thm:CoeffAlg})
with an upper bound on the degree of an
annihilating polynomial of an algebraic series specified by a system
$\mathcal S$ of polynomial equations.  The latter bound is singly
exponential in the size of the system $\mathcal S$, and is obtained using results about
quantifier elimination over the theory of real-closed fields.
The paper~\cite[Section 5]{Litow01} states a doubly exponential degree bound on the annihilating polynomial,
and leaves open the existence of a singly exponential bound.

A third natural computational problem on algebraic series is {\FinAlg}:
\emph{given a proper system of polynomial equations, determine whether denoted series has finite support}.
To the best of our knowledge, the theoretical complexity of \FinAlg~has not been studied before.
Concerning this problem we prove:

\begin{restatable}{theorem}{thmFinAlg} \label{thm:FinAlg}
The complement of {\FinAlg} is non-deterministic polynomial-time reducible to {\CoeffSLP},
and thereby {\FinAlg}  lies in $\coNP^{\PP}$.
\end{restatable}

The problems {\EqAlg} and {\FinAlg} have a natural characterisation in terms of context-free grammars.  Associated with a context-free grammar
over a $k$-element alphabet we have the \emph{census generating function} (see~\cite{Panholzer05})
\[ f(x_1,\ldots,x_k) := \sum_{\vec{v}\in\N^k} a_\vec{v} x_1^{v_1}\cdots x_k^{v_k} \, , \]
where $a_\vec{v}$ is the number of leftmost derivations that produce a word with Parikh image $\vec{v}$.  Then {\EqAlg} is polynomial-time equivalent to the problem of whether two grammars have identical census generating functions, whereas {\FinAlg} is polynomial-time equivalent to the problem of whether two grammars have respective census generating functions that differ in only finitely many entries.
Specialising \EqAlg~to the case of unary grammars we obtain the following corollary of~\cref{thm:AlgIT}.
\begin{corollary}
    \label{cor:unary CFG}
    Multiplicity equivalence of unary context-free grammars is in $\coRP^{\PP}$.
\end{corollary}

Some consequences of the corollary above are worth mentioning. 
First, the universality problem for (not necessarily unary) unambiguous context-free grammars is in $\coRP^{\PP}$,
since it reduces to multiplicity equivalence over a unary alphabet.
This problem was previously known to be in \pspace~(\cite[Theorem 2]{FJKW} and \cite[Theorem 10]{Clemente:EPTCS:2020}).
In turn, by the reductions in \cite[Theorems 8 and 9]{Clemente:EPTCS:2020}
the following two problems are also in $\coRP^{\PP}$:
deciding language inclusion of
1) a nondeterministic finite automaton in an unambiguous context-free grammar, and 
2) a nondeterministic context-free grammar in an unambiguous finite automaton.
The last two problems were known to be in \pspace~\cite{Clemente:EPTCS:2020}.

Corollary~\ref{cor:multiplicity equivalence-both}, below, generalises \cref{cor:unary CFG} from the unary case
to the more general case of letter-bounded context-free languages 
and context-free languages restricted to a given bounded language; we formally define these classes of languages in \cref{sec:cfg}.
In both cases the proof goes via a deterministic polynomial-time reduction to {\EqAlg},
and thus to {\DegSLP} by \cref{thm:AlgIT}.

%

\begin{corollary}
    \label{cor:multiplicity equivalence-both}
   The following two problems lie in $\coRP^{\PP}$:
   \begin{enumerate}
   \item Multiplicity equivalence of context-free grammars recognising  letter-bounded languages;
   \item Multiplicity equivalence of context-free grammars 
    restricted to a given bounded language $w_1^*\cdots w_k^*$.
   \end{enumerate}
\end{corollary}


\subsection{Related Work}
\label{sec:related}

Hensel's Lemma is based on Newton iteration applied to the ring of multivariate power series.
However, while the classical analysis of Newton's iteration provides convergence bounds  with respect to the usual Euclidean metric of $\R^k$,
Hensel's Lemma gives bounds in terms of an ultrametric that is more suitable to the context of power series.
Thus our approach to is related to the
scheme of so-called Newtonian program analysis~\cite{EsparzaKL10},  which involves the use of Newton's method to solve polynomial equations in a variety of different semirings.
The method has been applied to interprocedural
dataflow analysis as well as to
computing termination and reachability probabilites in
quasi-birth-death processes \cite{EtessamiWojtczakYannakakis:PE:2010}, multi-type branching processes \cite{EtessamiStewartYannakakis:STOC:2012}, and stochastic grammars \cite{EtessamiYannakakis:JACM:2009,EtessamiStewartYannakakis:ICALP:2013}.

Aside from language theory, algebraic series are an essential tool in combinatorics, where they are used to derive growth estimates on various types of combinatorial objects.   In this context, Flajolet and Soria~\cite[Theorem~1]{BanderierD15} have developed an explicit formula for the $n$-th coefficient of an algebraic power series.  A disadvantage of this formula 
for obtaining complexity bounds for {\CoeffAlg} is that it requires first computing an annihilating polynomial for the series, and when algebraic series are succinctly encoded as \emph{systems} of polynomial equations
annihilating polynomials require exponential degree in general.
Note also that if one has to hand the holonomic recurrence satisfied by the sequence of coefficients of an algebraic series, one can compute the coefficients one-by-one.   
However, even setting aside the overhead of computing such a recurrence, it is not clear to us whether one can prove Theorem~\ref{thm:CoeffAlg} via this route.


\section{Background}
\subsection{Complexity Theory}
\label{sec:complexity}
We briefly summarise some relevant notions from complexity theory (see~\cite[Chapter 7]{arora-barak} for more details).  
The class $\FPt$ is the function problem version of the decision problem class $\Pt$.
Let
$\Sigma$ be a finite alphabet.  The class $\sharpP$ is the
collection of functions ${f:\Sigma^* \rightarrow \N}$ for which there is
a non-deterministic polynomial-time Turing machine $M$ such that
$f(x)$ is the number of accepting computations of $M$ on input $x$.
The class $\PP$ (probabilistic polynomial time) is a decision analog
of $\sharpP$.  A language $L\subseteq\Sigma^*$ is in $\PP$ if there is a
non-deterministic polynomial-time Turing machine $M$ such that $x \in
L$ if and only if on input $M$ at least one half of the computations
of $M$ on input $x$ are accepting.  The \emph{counting hierarchy} is
the family of complexity classes inductively defined by
$\boldsymbol{C}_0 := \Pt$ and $\boldsymbol{C}_{k+1} :=
\PP^{\boldsymbol{C}_k}$.  In particular, the complexity class
$\coRP^{\PP}$ is included in $\boldsymbol{C}_2$: the second level of the
counting hierarchy.  We have $\bigcup_{k} \boldsymbol{C}_k \subseteq \pspace$.

\subsection{Polynomials and Arithmetic Circuits}
Let $X:=(x_1,\ldots,x_k)$ be a tuple of indeterminates.
Given $\vec{v}=(v_1,\ldots,v_k)$ in $\mathbb{N}^k$, we denote by $X^{\vec{v}}$
the monomial~$x_1^{v_1} \cdots x_k^{v_k}$. The \emph{total degree} of
$X^{\vec{v}}$ is defined to be~$|\vec{v}|:=\sum_{i=1}^k v_i$.  We
denote by $\mathbb{Z}[X]$ the (commutative) ring of polynomials with integer
coefficients and (commuting) variables in $X$.
The total degree of a polynomial is the maximum total degree of its constituent monomials.
The \emph{size} of a polynomial is the number of bits required to represent it
when its coefficients and degrees are written in binary.

An \emph{arithmetic circuit} over variables~$X$ is a directed acyclic graph with input gates labelled with the constants $0,1$ or with the indeterminates~$x_i$.  Internal gates are labelled with
one of the operations $+,-,\times$ and there is a distinguished output
gate.  Each gate of such a circuit represents polynomial in
$\Z[X]$ that is computed in an obvious bottom-up manner starting from the input gates.

The \emph{size} of a circuit is the number of its gates; see Figure~\ref{fig:circuit}.
The following proposition shows that circuits can be exponentially more succinct than polynomials.
\begin{restatable}{proposition}{propgpcircuit} \label{prop:gpcircuit}
Given $m \in \mathbb{N}$, there is a circuit of size $O(\log{m})$ that represents the polynomial~$\sum_{i=0}^{m} x^i$.
\end{restatable}
\noindent
Conversely, the total degree and the bit-length of the coefficients of a polynomial represented by a circuit
is always at most exponential in the size of the circuit.
%
The following is well-known~\cite{berkowitz, mahajan-vinay}:

\begin{proposition}
\label{prop:det-circuit}
Let $m \in \mathbb{N}$. Given an $m \times m$ matrix whose entries are distinct indeterminates there is an algorithm that runs in $\poly(m)$ time and produces
an algebraic circuit that represents the determinant of the matrix.
\end{proposition}

 A \emph{straight-line program} (SLP) is a sequence of
instructions corresponding to the sequential evaluation of an
arithmetic circuit.
In this paper, we treat arithmetic circuits and  SLPs as synonymous. 
 The following computational problems  for arithmetic circuits are well-studied, 
 see~\cite{kayal-saha, KoiranPerifel:TCS:2007, AllenderBurgisserKjeldgaard-PedersenMiltersen:JC:2009}.
Unless otherwise stated, all integers are represented in binary
and all polynomials are multivariate.
\begin{itemize}

\item {\EqSLP}: Given an arithmetic circuit computing a polynomial~$f$, decide whether $f$ is the zero polynomial.\footnote{Our usage here is non-standard in that {\EqSLP} typically refers to the problem of determining zeroness of an arithmetic circuit that represents an integer.  However the different versions of the problem are interreducible, so the distinction is not significant.}

\item {\DegSLP}: Given a positive integer $d$ and an arithmetic circuit computing
a polynomial~$f$, decide whether the total degree of~$f$ is at most~$d$.

\item {\CoeffSLP}: Given 
a multi-index $\vec{v}$ and prime $p$ (both encoded in binary), and
an arithmetic circuit computing a polynomial $f(X)$,
compute the residue modulo~$p$ of the coefficient of the monomial $X^{\vec{v}}$ in~$f$. 

\end{itemize}

We have the following 
reductions between these problems~\cite{kayal-saha, AllenderBurgisserKjeldgaard-PedersenMiltersen:JC:2009}:
\[ \mathsf{EqSLP} \leq_m \mathsf{DegSLP} \leq_r \mathsf{CoeffSLP} \, \]
where $\leq_m$ denotes a polynomial-time many-one reduction and $\leq_r$ denotes 
a randomized polynomial-time reduction. 

It is known that 
{\CoeffSLP} is $\sharpP$-hard (see for instance \cite{AllenderBurgisserKjeldgaard-PedersenMiltersen:JC:2009})
and from the proof of \cite[Theorem 4.1]{kayal-saha} it can be shown to be in $\FPt^{\sharpP}$.
Meanwhile 
{\DegSLP} is in  $\coRP^{\PP}$ \cite[Theorem 1.5]{kayal-saha},
but is not known to be $\NP$-hard.

\begin{figure}[t]
\begin{center}	
\begin{tikzpicture} 

\node[draw=none] at (2,6.5) (out) {$x^{2^n}+x^{2^n-1}+\ldots+x+1$};

\node[draw] at (2,5.5) (B6) {$+$};
\node[draw, label={left:$x^{2^n}$~}] at (1,5) (A61) {$\times$};
\node[draw, label={right:~~$x^{2^n-1}+\ldots+x+1$}] at (3,5) (A62) {$\times$};
\draw[<-]  (B6) edge (A61) ;
\draw[<-]  (B6) edge  (A62) ;
\draw[<-]  (out) edge  (B6) ;

\node[draw=none] at (1,4) (A51) {$\vdots$};
\node[draw=none] at (3,4) (A52) {$\vdots$};

\node[draw, label={left:$x^8$~~}] at (1,3) (A41) {$\times$};
\node[draw, label={right:~~$x^7+x^6+\ldots+x+1$}] at (3,3) (A42) {$+$};

\node[draw] at (2,2.5) (B3) {$\times$};
\node[draw, label={left:$x^4$~~}] at (1,2) (A31) {$\times$};
\node[draw, label={right:~~$x^3+x^2+x+1$}] at (3,2) (A32) {$+$};
\draw[<-]  (B3) edge (A31) ;
\draw[<-]  (B3) edge  (A32) ;
\draw[<-,  bend left=30]  (A41) edge (A31) ;
\draw[<-,  bend right=30]  (A41) edge (A31) ; 
\draw[<-]  (A42) edge (B3) ;
\draw[<-]  (A42) edge  (A32) ;

\node[draw] at (2,1.5) (B2) {$\times$};
\node[draw, label={left:$x^2$~~}] at (1,1) (A21) {$\times$};
\node[draw, label={right:~~$x+1$}] at (3,1) (A22) {$+$};
\draw[<-]  (B2) edge (A21) ;
\draw[<-]  (B2) edge  (A22) ;
\draw[<-,  bend left=30]  (A31) edge (A21) ;
\draw[<-,  bend right=30]  (A31) edge (A21) ; 
\draw[<-]  (A32) edge (B2) ;
\draw[<-]  (A32) edge  (A22) ;

\node[draw] at (2,.5) (B1) {$\times$};
\node[draw] at (1,0) (A11) {$x$};
\node[draw] at (3,0) (A12) {$1$};
\draw[<-]  (B1) edge (A11) ;
\draw[<-]  (B1) edge  (A12) ;
\draw[<-,  bend left=30]  (A21) edge (A11) ;
\draw[<-,  bend right=30]  (A21) edge (A11) ; 
\draw[<-]  (A22) edge (B1) ;
\draw[<-]  (A22) edge  (A12) ;

\draw[<->]  (0,5) edge node[midway,left]{$n$ times} (0,1) ;
 
\end{tikzpicture}
	\end{center}
	\caption{An arithmetic circuit
	 representing the polynomial~$\sum_{i=0}^{2^n} x^i$ of size~$O(n)$. This is a special case of \Cref{prop:gpcircuit} applied to $m = 2^n$.}
	 \label{fig:circuit}
\end{figure}
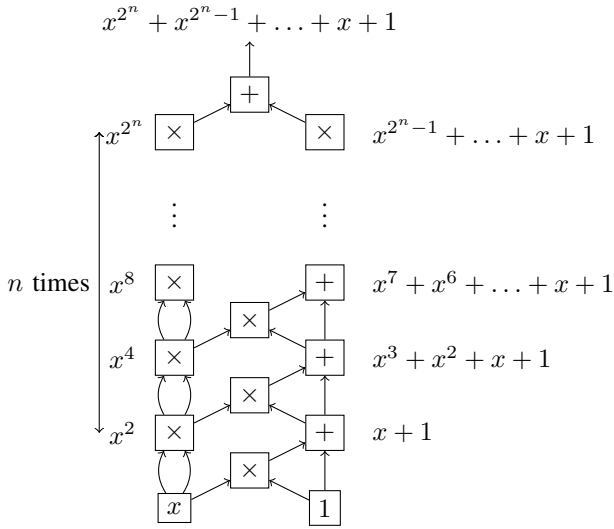

\subsection{Formal Power Series and Algebraic Equation Systems}
The ring of multivariate \emph{formal power series}
with integer coefficients over indeterminates $X=(x_1,\ldots,x_k)$ is defined by
\[ \PowS{\mathbb{Z}}{X} := \Big \{ \sum_{\vec{v}\in \mathbb{N}^k}
  a_{\vec{v}} X^{\vec{v}}\mid a_{\vec{v}} \in \mathbb{Z} \Big\}. \] 
  We denote by $\mathfrak{m}$ the ideal of the ring
$\PowS{\mathbb{Z}}{X}$ generated by $x_1,\ldots,x_k$, i.e., the ideal of all power series
with no constant term. 
For instance,  $x_2 + x_3^2 \in \mathfrak m$, 
but $1 + x_1 \not \in \mathfrak m$.
In language theory such series are
traditionally called \emph{quasiregular}~\cite[Chapter IV]{SalomaaSoittola}.
We note that the units of $\PowS{\mathbb{Z}}{X}$, that is, the elements with multiplicative inverses, are those of the form $\pm 1+ f$ where~$f\in \mathfrak{m}$.

A series $A \in \, \PowS{\mathbb{Z}}{X}$ is said to be \emph{algebraic}
if there exist polynomials $p_0,\ldots,p_d \in \Z[X]$, not all zero, such that 
\[ p_0 + p_1 A + p_2A^2 + \ldots + p_dA^d = 0. \]
We call the polynomial $\sum_{n=0}^d p_n y^n \in \Z[X][y]$ an
\emph{annihilating polynomial} of $A$.

Let $Y=(y_1,\ldots,y_\ell)$ be a tuple of variables.  
A
\emph{polynomial equation system} $\mathcal{S}$ over the
indeterminates~$X$ and variables~$Y$ consists of a collection of $\ell$
equations
\begin{align*}
y_i=P_i	\qquad (i\in \{1,\ldots,\ell\})
\end{align*}
where $P_i\in \Z[X][Y]$.  In computational problems we assume that the polynomials
$P_i$ are presented as lists of monomials and integers are written in binary.
The \emph{size} of $\mathcal{S}$ is the sum of sizes of its polynomials $P_1, \dots, P_\ell$; in particular, it is at least the number $k$ of indeterminates and $\ell$ of variables.
In case the coefficients of the $P_i$ are all polynomials with nonnegative integer coefficients from $\mathbb N$,
we say that $\mathcal{S}$ is \emph{defined over $\N$}.

A vector~$\vec{A}\in \PowS{\Z}{X}^{\ell}$ of formal power series
evaluating each polynomial $P_i$ at~$\vec{A}$
yields an identity $A_i=P_i(\vec{A})$ of formal power series.
%

We introduce a condition guaranteeing unique solutions of a certain kind.
We say that the system~$\mathcal{S}$ is \emph{proper}
if for each polynomial~$P_i=\sum a_{\vec{v}} Y^{\vec{v}}$,
we have $a_{\vec{v}} \in \mathfrak{m}\cap \, \Z[X]$
whenever $|\vec{v}|\leq 1$.
In other words, the coefficients of monomials of total degree at most one are in the ideal $\mathfrak m$.

Inductively define the \emph{approximating sequence} 
  $\boldsymbol{A}^{(0)}, \boldsymbol{A}^{(1)}, \ldots$ 
  in $\PowS{\Z}{X}^{\ell}$ by 
  $\boldsymbol{A}^{(0)} := \boldsymbol{0}$ and 
\begin{align}
    \label{eq:approx}
    \boldsymbol{A}^{(n+1)} :=
    (P_1(\boldsymbol{A}^{(n)}),\ldots,
    P_\ell(\boldsymbol{A}^{(n)}))
\end{align}
for all $n\in\mathbb{N}$.
When applied to a proper system,
every component in this sequence converges pointwise to a vector of power series $\vec{A} \in \PowS{\Z}{X}^{\ell}$
 with respect to the product topology of $\PowS{\mathbb{Z}}{X}$.
This vector $\vec A$ is a solution of the system,
called the \emph{strong solution}~\cite[Section IV.1, Theorem 1.1]{SalomaaSoittola}.
Moreover, the strong solution is the unique quasiregular solution,
and each component thereof is algebraic~\cite{Panholzer05}.
We say that the first
component of this unique solution is \emph{the formal power series computed by~$\mathcal{S}$}.
{More generally, if for an arbitrary system $\mathcal S$ the approximating sequence \eqref{eq:approx} converges to $\vec{A}$, then the quasiregular part of $\vec{A}$ (which is obtained by omitting the term of degree zero) is the solution of a proper polynomial system~\cite{KuichSalomaa}.

\begin{example}
\label{ex:proper-system}
The polynomial equation $y=x+x^2-2xy+y^2$ is proper with two  solutions~$x$ and $1+x$, of which only $x$ is quasiregular.
The approximating sequence in this case is $A_0 = 0$ and $A_n = x + x^{2^n}$, which can be proved by induction.
%
%
In particular, $A_n$ converges to $x$ as required.

 The  generating function $C(x)$ of the Catalan numbers is a solution of the polynomial equation $y = 1 + xy^2$, which is not proper.
 However,  its quasiregular part $C(x)-1$ is the solution of the proper system $y= x + 2xy + xy^2$.
 \end{example}

In analogy with the problems {\EqSLP}, {\DegSLP}, and {\CoeffSLP} for algebraic circuits,
we consider in this paper the following problems on algebraic power series.
\begin{itemize}
    \item {\EqAlg}: Given a proper system of polynomial equations with unique quasiregular solution $\vec A$,
    decide whether the first component $A_1$ of $\vec A$ is the zero power series.
    \item {\FinAlg}: Given a positive integer $d$ and a proper system of polynomial equations with unique quasiregular solution $\vec A$,
    decide whether the total degree of the first component $A_1$ of $\vec A$ is finite.
    \item {\CoeffAlg}: Given a multi-index $\vec{v}$, a prime $p$ (both encoded in binary),
    and a proper system of polynomial equations with unique quasiregular solution $\vec A$,
    compute the residue module~$p$ of the coefficient of the monomial $X^{\vec{v}}$ in the first component $A_1$ of $\vec A$. 
\end{itemize}




\subsection{Hensel's Lemma}
\label{sec:hensellemmam}
Let $R$ be a commutative ring with unity.  A \emph{valuation}
on $R$ is a map $v:R\rightarrow \mathbb{N}\cup\{\infty\}$ such that
for all $x,y \in R$ we have:
\begin{enumerate}
\item $v(x)=\infty$ iff $x=0$,
\item $v(x+y)\geq \min\{v(x),v(y)\}$,
\item $v(xy)=v(x)+v(y)$.
\end{enumerate}
It is easy to check that, if $v$ is a valuation on $R$,
then the function $d(x,y):=2^{-v(x-y)}$ (with the convention that $2^{-\infty}=0$)
defines an ultrametric on $R$.
We say that $R$ is \emph{complete} with respect to~$v$
if it is a complete metric space with respect to~$d$ (in the standard sense).


We now state a version of Hensel's Lemma
that is convenient for our purposes. 
This combines the multivariate Hensel Lemma found in
~\cite[Section~4.6, Theorem~2]{Bourbaki85} and~\cite[Exercise 7.26]{Eisenbud95}
with an assertion of quadratic convergence.
We will use Hensel's Lemma in \cref{sec:rational} to define a sequence of rational approximations
converging to the quasiregular solution of a proper system of polynomial equations.
The proof is a straightforward adaptation of classical arguments from the literature,
however we include a proof for the convenience of the reader. 
 
Assume that the ring $R$ is
complete with respect to a valuation $v$.  Let
$f_1,\ldots,f_\ell$ lie in the polynomial ring $R[Y]$, where $Y=(y_1,\ldots,y_\ell)$ is a tuple of
distinct indeterminates.  Recall that the \emph{derivative matrix} of
$\vec{f}=(f_1,\ldots,f_\ell)$ is
\begin{align}
\label{eq-der-mat-f}
    D\vec{f} := \left( \frac{\partial f_i}{\partial y_j} \right)_{1 \leq i,j \leq \ell}  \in R[Y]^{\ell \times \ell}
\end{align}
and its \emph{Jacobian} is $J_{\vec{f}} := \det (D \vec{f}) \in R[Y]$.

For a positive integer $\ell$, we extend $v$ to a map
$v:R^\ell \rightarrow \mathbb{N}$ by writing
$v(x_1,\ldots,x_\ell):=\min \{ v(x_i):1\leq i \leq \ell\}$.
Note that the extension of $v$ to $R^\ell$ is not in general a valuation.

\begin{theorem}[Hensel's Lemma]
\label{thm:hensel}
  Let $\vec{a} \in R^\ell$ be such that $v(\vec{f}(\vec{a}))>0$ and
  $J_{\vec{f}}(\vec{a})$ is a unit in $R$.
Consider the sequence $(\vec{a}_n)_{n=0}^\infty$, defined inductively
by $\vec{a}_0=\vec{a}$ and
\begin{gather} \vec{a}_{n+1} = \vec{a}_n - (D\vec{f}(\vec{a}_n))^{-1}
  \vec{f}(\vec{a}_n) \qquad (n\in\mathbb{N}) \, .
  \label{eq:RECUR}
  \end{gather}
Then $(\vec{a}_n)_{n=0}^\infty$ is a well-defined sequence in $R^\ell$ 
and there exists 
$\vec{\alpha} \in R^\ell$ such that $\vec{f} (\vec{\alpha}) = 0$ and $v(\vec{\alpha}-\vec{a}_n)  \geq 2^{n}$ for all~$n\in \mathbb{N}$.
\end{theorem}
\noindent The statement of well-definedness refers to the
claim that the matrix $D\vec{f}(\vec{a}_n)$ is
invertible for all $n$,
which amounts to the fact that
$J_{\vec{f}}(\vec{a}_n)$ is a unit of $R$ for all $n$.
\begin{proof}
  We will need the following elementary fact about complete rings (see, for example,~\cite[Lemma 2.6]{Dries}):
  If $x\in R$ is a unit and
  $v(x-y)>0$ then $y$ is also a unit.

  We will show by induction on $n\in\mathbb{N}$  that the sequence
  $(\vec{a}_n)_{n=0}^\infty$ is well-defined and satisfies the
  following for all $n$:
  \begin{enumerate}
  \item $v(\vec{f}(\vec{a}_n)) \geq 2^n$,
  \item $v(\vec{a}_n-\vec{a}_{n-1}) \geq 2^{n-1}$ if $n>0$,
    \item $J_{\vec{f}}(\vec{a}_n)$ is a unit.
    \end{enumerate}

    Suppose first that $n=0$.  Then, Item 2 holds vacuously, while, by the
    choice $\vec{a}_0=\vec{a}$, Items 1 and 3 are hypotheses of the
    theorem.

For the induction step, by the multivariate Taylor's Theorem, for
$\vec{\varepsilon}=(\varepsilon_1,\ldots,\varepsilon_\ell) \in R^\ell$ it
holds that
\begin{gather} \vec{f}(\vec{a}_n+\vec{\varepsilon}) =
  \vec{f}(\vec{a}_n) + 
  ((D\vec{f})(\vec{a}_n)) \vec{\varepsilon}
+ \sum_{1\leq i,j\leq \ell}
\varepsilon_i\varepsilon_j \vec{c}_{ij} \, ,
\label{eq:TAYLOR}
\end{gather}
where $\vec{c}_{ij}\in R^\ell$ for all $1\leq i,j\leq \ell$.  By Item~3 
 the matrix $(D\vec{f})(\vec{a}_n)$ has an inverse that lies in 
$R^{\ell \times \ell}$
%
Thus we may put
$\vec{\varepsilon}:=-((D\vec{f})(\vec{a}_n))^{-1} \vec{f}(\vec{a}_n) \in R^\ell$,
in which case
$v(\vec{\varepsilon}) \geq v(\vec{f}(\vec{a}_n)) \geq 2^n$.
The first inequality follows from the general fact that $v(A \cdot \vec b) \geq v(\vec b)$ for every matrix $A \in R^{\ell \time \ell}$ and vector~$\vec b \in R^{\ell}$.
%
By Equations~\eqref{eq:RECUR} and~\eqref{eq:TAYLOR},
\[ \vec{f}(\vec{a}_{n+1}) = \vec{f}(\vec{a}_n+\vec{\varepsilon}) =
  \sum_{1\leq i,j\leq \ell} \varepsilon_i\varepsilon_j \vec{c}_{ij} \,
  . \] We conclude that $v(\vec{f}(\vec{a}_{n+1})) \geq 2^{n+1}$, as
required in Item~1.  We also have
$v(\vec{a}_{n+1}-\vec{a}_n) = v(\vec{\varepsilon}) \geq 2^n$, as
required in Item~2.

Regarding Item~3, we will use the following claim (proved in \cref{sec:extendedPre}).
\begin{restatable}{claim}{claimvalpoly}
    \label{claim:valuation and polynomials}
    For a univariate polynomial $p(x) \in R[x]$ and two elements $a, b \in R$,
    $v(p(a) - p(b)) \geq v(a - b)$.
\end{restatable}
\noindent
Since $J_{\vec{f}}$ is a polynomial in $R[Y]$,
by \cref{claim:valuation and polynomials} componentwise,
it follows that 
\[ v(J_{\vec{f}}(\vec{a}_{n+1}) - J_{\vec{f}}(\vec{a}_n)) \geq v(\vec{a}_{n+1} - \vec{a}_n) > 0\] and
hence, since $J_{\vec{f}}(\vec{a}_n)$ is a unit,
by the observation right at the start of the proof,
$J_{\vec{f}}(\vec{a}_{n+1})$ is a unit as well.
This establishes Item~3 and the induction is complete.
Note also that this implies that the sequence $(\vec{a}_n)_{n=0}^{\infty}$ is well-defined, and by construction, for all~$n$, $\vec{a}_n \in R^\ell$.

By Item~2 we have $v(\vec{a}_m - \vec{a}_n) \geq 2^n$ for all $m > n$.
Hence the sequence $(\vec{a}_n)_{n=0}^\infty$ is Cauchy and so,
by completeness of $R$, it converges to a limit $\vec{\alpha} \in R^\ell$.
But by Item~1 and continuity of $\vec{f}$ we have
$\vec{f}(\vec{\alpha})=0$.  Furthermore, by continuity of $v$ we have
$v(\vec{\alpha}-\vec{a}_n) \geq 2^n$.  
  \end{proof}

\begin{remark}
  \label{rem:CONVENIENT}
  One often finds formulations of Hensel's Lemma that require the ring
  $R$ to be a \emph{valuation ring}, i.e., the subring of a valued
  field comprising all elements having non-negative valuation.  Examples
  of valuation rings are the $p$-adic integers $\mathbb{Z}_p$ (which is the
  valuation ring of the field $\mathbb{Q}_p$) and the ring $F[[x]]$ of
  univariate power series over a field $F$ (which is the valuation ring of the
  field $F((x))$ of Laurent series).  However, as we note in \Cref{sec:coeffAlg},
  Theorem~\ref{thm:hensel} applies to the ring of multivariate power
  series $\mathbb{Z}[[X]]$, which is not a valuation ring.
\end{remark}

\section{Computational Complexity of \texorpdfstring{{\CoeffAlg}}{CoeffAlg}}
\label{sec:coeffAlg}
Recall that an instance  of the {\CoeffAlg} problem comprises 
a proper polynomial equation system  computing a formal power series~$A$, a monomial $X^{\vec{v}}$, and  a prime $p$. 
The problem asks to compute the residue modulo~$p$ of the coefficient of~$X^{\vec{v}}$ in $A$.  
In analysing the complexity of the problem we assume that all integers are represented in binary.
In this section we prove the following theorem:

\thmCoeffAlg*

The proof is spread across several subsections.
Throughout we work with a tuple~$X=(x_1,\ldots,x_k)$ of commuting indeterminates and a proper polynomial equation system $\mathcal{S}$ over 
 the indeterminates~$X$ and variables~$Y=(y_1,\ldots,y_\ell)$---forming part of the input to the $\mathsf{CoeffAlg}$ problem.
We let~$s$ denote the length of the description of~$\mathcal{S}$.

Let $\vec{A}=(A_1,\ldots,A_\ell) \in \PowS{\Z}{X}^\ell$ be the unique quasiregular solution of~$\mathcal{S}$ in formal power series. 
Roughly speaking our reduction of {\CoeffAlg} to {\CoeffSLP} 
involves constructing a 
sufficiently close polynomial approximation of $A_1$ that admits an efficiently computable representation as a circuit. 
To do this we will apply Hensel's Lemma to the power series ring~$\Z[[X]]$.

\subsection{A Valuation on the Ring of Power Series}
In the following we denote by $R$ the ring of formal power series $\PowS{\Z}{X}$ and by $R_0$ the subring $\Z[X]$ of polynomials.
Recall that $\mathfrak{m}$ is the ideal in $R$ generated by 
$x_1,\ldots,x_k$.  Given~$g \in \mathfrak{m}$, the element
$1-g$ is a unit in $R$, having inverse $\sum_{n=0}^\infty g^n$.
(The latter sum converges as a power series
by virtue of the fact that $g$ is in $\mathfrak m$.)
Indeed, the units in $R$ are precisely those elements~$f$ such that $\pm f$ has the above form.  An element of $R$ is said to be \emph{rational} if it has the form $fg^{-1}$, where $f,g \in R_0$ and $g$ is a unit.%
\footnote{Since $\Z$ is a so-called Fatou ring~\cite[Chapter 7]{BerstelR},
the rational elements of $R$ according to the above definition are precisely those lying in $\Q(X)\cap R$, where $\Q(X)$ is the field of rational functions over indeterminates $X$.}

We define a map $\ord:R \to \N \cup\{\infty\}$  
by $\ord(0):=\infty$ and otherwise
\[\ord \left( \sum_{\vec{v}\in \N^k} a_{\vec{v}} X^{\vec{v}}\right) :=\min \{|\vec{v}| : a_{\vec{v}}\neq 0\} \, . \]
It is easily seen that the map $\ord$
is  a valuation:
for all $f,g \in R$
\[
\ord(f+g) \geq \min\{\ord(f),\ord(g)\},
\]
and 
\[ 
\ord(fg) = \ord(f)+\ord(g).
\]
Since elements with a strictly positive valuation are precisely the elements without a constant term,
we recover the ideal $\mathfrak{m}$ as the set $\{ f \in R : \ord(f) > 0 \}$.

It is standard that $R$ is a complete ring  with respect to this valuation; see~\cite[Section 7.1]{Eisenbud95}.
For $\ell \in \N$, we extend the map $\ord$ to vectors $R^\ell$ of power series by,
for all $\vec{B} \in R^\ell$,
\[\ord(\vec{B}) := \min \{
 \ord(B_i): 1 \leq i \leq \ell\}.\]

 \subsection{Approximation by Rational Series}
 \label{sec:rational}
For our fixed polynomial system $\mathcal{S}$, given by the equations 
\begin{align*}
y_i=P_i	\qquad (i\in \{1,\ldots,\ell\})
\end{align*}
 we write $f_i:= y_i - P_i\in R_0[Y]$. We denote by $\vec{f}$ the  vector of polynomials 
\begin{equation}\label{eq:system_f}
	\vec{f} = (f_1,\dots,f_\ell) \, .
\end{equation}
Clearly $\vec{B} \in R^\ell$ is a solution of $\mathcal S$ if and only 
$\vec{f}(\vec{B})=0$.  Since $\mathcal S$ is  proper 
it follows that $\vec{f}$ has a unique quasiregular zero, namely the  unique quasiregular solution $\vec{A}$ of~$S$.

 Our goal is to apply~\Cref{thm:hensel}---the multivariate Hensel's Lemma---to define a sequence of approximations of the unique quasiregular zero~$\vec{A}$ of $\vec{f}$.  We have already established that $R$ is complete with respect to the valuation $\mathrm{ord}$.  To apply Hensel's Lemma it remains to find a suitable initial value to start the iteration, namely,
 $\vec{a} \in R^\ell$ such that 
 $J_{\vec{f}}(\vec{a})$ is a unit in $R$ and $\ord(\vec{f}(\vec{a})) > 0$.
 The following two claims show that any $\vec{a} \in \mathfrak{m}^\ell$ will do.
Recall that $1+\mathfrak{m}$ denotes the set of elements of the form $1 + f$ with $f \in \mathfrak m$.
%
 \begin{restatable}{claim}{propmultivarhenseljacobian} \label{prop:multivar_hensel_jacobian}
$J_\mathbf{f}(\vec{a})\in 1+\mathfrak{m}$
for all  $\vec{a}\in \mathfrak{m}^{\ell}$.
\end{restatable}
We give a sketch proof of the claim;
see \Cref{sec:app-proofs-sec-iii} for the full proof.
Assume that $\vec{a} \in \mathfrak{m}^\ell$.  Using the fact that 
$\mathcal S$ is proper, one can show that the diagonal entries of
the derivative matrix $D\vec{f}(\vec{a})$ all lie in $1+\mathfrak{m}$
whereas the off-diagonal entries all lie in~$\mathfrak{m}$.
In other words, the matrix $D\vec{f}(\vec{a})$ is elementwise congruent to the identity matrix modulo $\mathfrak{m}$.  Since the determinant of a matrix is a polynomial function of its entries, $J_{\bf f}(\vec{a}) \in 1+\mathfrak{m}$.

The fact that $\ord(\vec{f}(\vec{a}))>0$ for all
$\vec{a}\in\mathfrak{m}^\ell$ follows immediately from
\begin{restatable}{claim}{propmultivarhenselnorm} \label{prop:multivar_hensel_norm}
  For all $\vec{a}\in\mathfrak{m}^\ell$
  we have $\vec{f}(\vec{a}) \in \mathfrak{m}^\ell$.
\end{restatable}
The claim holds by observing that, since $\mathcal S$ is proper, the coefficient of $Y^{\vec{0}}$ in $f_i \in R_0[Y]$ lies in 
$\mathfrak{m}$. Again, see \Cref{sec:app-proofs-sec-iii} for details.




Claims~\ref{prop:multivar_hensel_jacobian} and~\ref{prop:multivar_hensel_norm} imply that choosing $\vec{a}$ to be $\vec{0}\in R^{\ell}$ satisfies the assumptions
of~\Cref{thm:hensel}.  Now define a sequence $(\vec{a}_n)_{n=0}^\infty$ by
${\bf a}_0 = {\bf 0}$ and
 \begin{eqnarray}
\nonumber
\vec{a}_{n+1} &=& \vec{a}_n - (D{\bf f})({\bf a}_n)^{-1} {\bf f}({\bf a}_n) \\
 \label{eq:update}
&=& {\bf a}_n -  \frac{1}{J_{\bf f}({\bf a}_{n})} \, 
 \mbox{Adj}(D{\bf f}({\bf a}_{n}))\, {\bf f}({\bf a}_n),
\end{eqnarray}
 where the entries of $D\vec{f}$, the derivative matrix of~$\vec{f}$, are in~$R_0[Y]$,
 $J_{\bf f}$ is the determinant of $D{\bf f}$, as defined in Equation~\eqref{eq-der-mat-f}. Here, 
 $\mbox{Adj}(D{\bf f}({\bf a}_{n}))$ is the adjugate matrix whose entries are just the cofactors of~$D{\bf f}({\bf a}_{n})$.
Applying~\Cref{thm:hensel} to ${\bf f}$,
for all $n \in \mathbb{N}$ we have 
\begin{gather}
	\ord(\vec{A} - \vec{a}_{n}) \geq 2^n \, .
\label{eq:convergence}
\end{gather}
The sequence $(\vec{a}_n)_{n=0}^\infty$ is the desired approximating sequence of $\vec{A}$. 

Note that by \cref{thm:hensel}(3), the Jacobian $J_{\bf f}({\bf a}_{n})$ is a 
unit in $R$ for all $n\in \mathbb{N}$.
Since the rational elements of $R$ (i.e., the rational power series) form a ring,
it follows by a straightforward induction on $n$, using the recurrence~\eqref{eq:update}, that each ${\bf a}_{n}$ is a vector of rational elements in $R$. Consequently, for all $i \in \{1,\ldots,\ell\}$ and $n \in \mathbb{N}$,
the component $\vec{a}_{n,i}$ can be written
for some polynomials~$g_{n,i}, h_{n,i}\in R_0$ as
\begin{align}
    \label{eq:a n i}
    {\bf a}_{n,i} =\frac{g_{n,i}}{1- h_{n,i}}, \quad \text{with } \ord(h_{n,i})\geq 1.
\end{align}


\begin{example}
\label{ex:iteration-quasiregular}
Consider again the polynomial system $y=x+x^2-2xy+y^2$, introduced in Example~\ref{ex:proper-system},
with the unique quasiregular solution~$x$. 
Write $f(y) = y -x-x^2+2xy-y^2$ for the polynomial defining the system
and $D_f(y) = 1 +2x - 2y$ for its derivative with respect to~$y$.
Then following the iterative procedure given by Hensel's Lemma,
for all $n\in \N$, we have 
\begin{equation}
\label{eq:example-rec}
    a_{n+1} = a_n - \frac{f(a_n)}{D_f(a_n)} = \frac{ - a_n^2+x^2+x}{1+2x-2a_n} \, .
\end{equation}  
 Starting the iteration with~$a_0 = 0$, we  prove by induction  that the term $a_n$, for $n\geq 1$, has the following closed-form expression:
 \begin{equation}
 \label{eq:example-closed-form}
     a_n= x- \frac{x^{2^n}}{(x+1)^{2^n}-x^{2^n}} \, .
 \end{equation}
The inductive proof is detailed in~\Cref{sec:app-example}.
We note that  $\ord(a_n-x)=2^n$, matching the bound in~\Cref{thm:hensel}. 


\end{example}

\subsection{Approximation by Polynomials}
\label{sec:polynomial}


In the previous section,
we have shown how to approximate the quasiregular solution $A$ by rational power series with good convergece speed, see~\eqref{eq:convergence}.
In this section we achieve the same convergence  by approximating \emph{even by polynomials}.
This will allow us to implement the polynomial approximants with algebraic circuits in \cref{sec:circuits}.
To this end, we define a sequence $\widetilde{\bf a}_{n} \in R_0^\ell$
of vectors of polynomials that satisfies the following convergence bound for all $n$
(see Equation~\eqref{eq:convergence}):
\begin{gather}
	\ord(A_i - \widetilde{\vec{a}}_{n,i}) \geq 2^n
\label{eq:convergence2}
\end{gather}
where $\widetilde{\vec{a}}_{n,i}$ is the $i$-th component of~$\widetilde{\vec{a}}$.
The definition of $\widetilde{\bf a}_{n,i}$ 
uses a classical technique of division elimination by Strassen~\cite{strassen-division}. While Strassen
obtains a polynomial written as a rational function,
we are computing a polynomial approximation of the power series
defined by a rational function.

%
Given the representation of $\vec{a}_{n,i}$ in Equation~(\ref{eq:a n i}),
we have
\[ {\bf a}_{n,i}  = 
g_{n,i}\sum_{j=0}^\infty h_{n,i}^j \, .\]
A polynomial approximant~$\widetilde{\bf a}_{n,i}$ can then be obtained by  truncating the above  infinite sum to the first~$2^n$ terms.  We thus write
\begin{align}
\label{eq-approx-def}
    \widetilde{\bf a}_{n,i}:= 
g_{n,i} \sum_{j=0}^{2^n-1} h_{n,i}^j \, . 
\end{align}
Then we have
\begin{eqnarray*}
    \ord ({\bf a}_{n,i} - \widetilde{\bf a}_{n,i})
    &=& \ord(g_{n,i} \sum_{j=2^n}^{\infty} h_{n,i}^j) \\
    &=&\ord(g_{n,i}) + \ord( \sum_{j=2^n}^{\infty} h_{n,i}^j ) \\
    &\geq & 2^n \quad \text{(since $\ord(h_{n,i})\geq 1$)}
\end{eqnarray*}
The desired bound in Equation~\eqref{eq:convergence2} now follows from~\eqref{eq:convergence}
and the strong triangle inequality (Property 2 of valuations in~\Cref{sec:hensellemmam}):
\begin{eqnarray*}
	\ord(A_i - \widetilde{\vec{a}}_{n,i}) &=& \ord(A_i - {\bf a}_{n,i} + {\bf a}_{n,i} - \widetilde{\vec{a}}_{n,i}) \\ &\geq & \min\{\ord(A_i - {\bf a}_{n,i}), \ord({\bf a}_{n,i} - \widetilde{\vec{a}}_{n,i}) \} \\
	&\geq & 2^n.
\end{eqnarray*}

%


\subsection{Approximation by Circuits}
\label{sec:circuits}
In this subsection we show that for all  $n\in \N$ 
the polynomial approximant~$\widetilde{\bf a}_{n}$ of~$\vec{A}$ can be represented by a circuit that can moreover be computed in time polynomial in $n$ and the size $s$ of the  equation system~$\mathcal{S}$.
We construct the circuit in two steps. 
First 
we show how to construct circuits for the two polynomials $g_{n,i},h_{n,i}\in R_0$
in Equation~(\ref{eq:a n i}) representing ${\bf a}_{n,i}$.

\begin{restatable}{claim}{propnewtonratfunction} \label{prop:newton_rat_function}
    There is an algorithm that, given an equation system~$\mathcal S$, 
    $i \in \{1,\ldots,\ell\}$, and $n \in \mathbb{N}$, 
    produces  circuits $C_{n,i}, D_{n,i}$ representing 
    polynomials $g_{n,i},$ $h_{n,i} \in R_0$ respectively,
    as in Equation~(\ref{eq:a n i}).
    The algorithm runs in $\poly(s,n)$ time and hence $C_{n,i}, D_{n,i}$ have size $\poly(s,n)$.
\end{restatable}

Next, we follow the procedure described in Section~\ref{sec:polynomial} to
construct a circuit for the approximants~$\widetilde{\bf a}_{n,i}$ obtained from $g_{n,i}$ and $h_{n,i}$.
Since we will need these approximants just for the first component $i = 1$,
we introduce this specialisation already in the next claim.

\begin{restatable}{claim}{proppolycircuit} \label{prop:poly_circuit}
There is an algorithm that, given an equation system~$\mathcal S$
and $n \in \mathbb{N}$,
produces circuit $E_n$ representing the polynomial~$\widetilde{\bf a}_{n,1}$, defined in Equation~\eqref{eq-approx-def} for $i = 1$.
The algorithm runs in $\poly(s,n)$ time, and $E_{n}$ has size $\poly(s,n)$.
\end{restatable}

\subsection{The upper and lower complexity bounds}
  
 We are now ready to conclude the proof of Theorem~\ref{thm:CoeffAlg}. By Equation~\eqref{eq:convergence2} and Claim~\ref{prop:poly_circuit}, for our fixed proper polynomial system~$\mathcal{S}$, for all~$n\in \N$, we compute a circuit~$E_{n}$ representing a polynomial approximant of the formal power series solution~$A$ computed by~$\mathcal{S}$.
 Indeed,  the polynomial represented by~$E_{n}$ agrees with~$A$   in all monomials with total degree at most~$2^n$. 
 Moreover, there is an algorithm computing~$E_n$ in   $\poly(s,n)$ time, and $E_{n}$ has size $\poly(s,n)$.
 It remains to observe that for the  input monomial~$X^{\vec{v}}$ and  prime~$p$,  the residue modulo~$p$ of the coefficient of ~$X^{\vec{v}}$ in~$A$ and $E_{\log(|\vec{v}|)+1}$ agree. This yields   a polynomial-time reduction from   {\CoeffAlg}  to \CoeffSLP. 

For the converse direction, note that {\CoeffSLP} is not trivially subsumed by {\CoeffAlg} as the latter requires the input system of equations to be proper.
The following claim, proven in \cref{sec:app-proofs-sec-iii}, outlines a straightforward reduction from {\CoeffSLP} to {\CoeffAlg}. 

\begin{restatable}{claim}{sharpPhardCoeffAlg} \label{claim:sharpPhardCoeffAlg}
There is a polynomial-time reduction from  {\CoeffSLP} to {\CoeffAlg}. 
\end{restatable}


\section{Complexity of 
\texorpdfstring{{\EqAlg}}{EqAlg} and \texorpdfstring{{\FinAlg}}{FinAlg}}

In this section we establish complexity upper bounds for {\EqAlg} and {\FinAlg}.  
A key ingredient behind these results is a new 
singly exponential upper
bound on the degree of an annihilating polynomial of the strong solution of a proper polynomial system. 

There are two different algorithms that given a proper polynomial system~$\mathcal{S}$ compute an annihilating polynomial of its strong
solution.  The algorithm of 
Kuich and Salomaa~\cite[Section~16]{KuichSalomaa} 
is based on
multiresultants and polynomial factorisation, whereas the algorithm of
Panholzer~\cite{Panholzer05} uses Gr\"{o}bner bases.  
It is noted 
in~\cite[Example~9]{Panholzer05} that one cannot always obtain
an annihilating polynomial of the strong solution merely by performing   
quantifier elimination on the system $\mathcal S$ over the first-order theory of algebraically closed fields.
This is because the elimination ideal may be trivial and does not provide any information on the strong solution.
In such a case further
work is needed to isolate the strong solution, such as decomposing the
variety of all solutions of $\mathcal S$ into its irreducible
components as in~\cite{Panholzer05}.

In this section we take an alternative approach.  We observe that the strong solution is first-order definable in the ordered field of Puiseux series and we use standard quantifier
elimination results for real closed fields to compute an annihilating
polynomial of the strong solution.  In particular, we obtain a singly exponential (in the size of $\mathcal S$)
upper bound on the degree of an annihilating polynomial.  
Paper~\cite{Litow01} states a doubly exponential upper bound on the degree of an annihilating polynomial, based on an analysis of the Kuich-Salomaa algorithm.

\subsection{Real Closed Fields and Puiseux Series}
We work with the first-order theory of real closed fields over the
language of ordered rings with constant symbols for $0$ and $1$~\cite[Chapter 2]{basu2006algorithms}.
Recall that a model of this theory is an ordered field in which the
intermediate value theorem holds for all polynomials (such as the field
of real numbers or the field of Puiseux series with real
coefficients, defined below).  Atomic formulas have the form
$P(x_1,\ldots,x_n) \sim 0$, where $P \in \mathbb{Z}[x_1,\ldots,x_n]$
and ${\sim}\in \{<,=\}$.  We say that a formula $\Phi$ is \emph{built over a
set of polynomials $\mathcal{P}$} if every polynomial mentioned in
$\Phi$ lies in $\mathcal{P}$.  It well-known that the theory of real
closed fields admits quantifier elimination.  Here we will use the
following quantitative formulation of quantifier elimination, which is
a specialisation of~\cite[Theorem 14.16]{basu2006algorithms}.

\begin{theorem}
Let $\mathcal{P}$ be a set of $s$ polynomials, each of degree at most $d$ and having coefficients of bit-size at most $\tau$. 
Given tuples $X=(x_1,\ldots,x_{k_1})$, $Y=(y_1,\ldots,y_{k_2})$ and 
$Z=(z_1,\ldots,z_{k_3})$ of first-order variables, consider the formula 
\[ \Phi(X):= \exists Y \forall Z \, \Psi(X,Y,Z) \, ,\] where $\Psi(X,Y,Z)$
is a quantifier-free formula built over $\mathcal{P}$.  Then there
exists an equivalent\footnote{Equivalent over every real closed field.}
quantifier-free formula $\Phi'(X)$ that is built
over a set of polynomials having degree bounded by $d^{ck_2k_3}$ and
coefficients of bit size bounded by $\tau d^{ck_1k_2k_3}$ for some absolute constant $c$.
\label{thm:elim}
\end{theorem}

Let $F$ be a field and
$X=(x_1,\ldots,x_k)$ a tuple of commuting indeterminates.
A \emph{Puiseux series} with coefficients in $F$ and variables $X$ is a formal series 
\begin{gather}
f:=\sum_{\vec{\alpha} \in \mathbb{Q}^k} c_{\vec{\alpha}} X^{\vec{\alpha}}
\label{eq:puiseux}
\end{gather}
whose \emph{support}
$S:=\{ \vec{\alpha}\in\Q^k : c_{\vec{\alpha}} \neq 0\}$ is well-ordered with respect to the lexicographic order on
$\mathbb{Q}^k$ and also satisfies $S\subseteq \frac{1}{q} \mathbb{Z}$ for some positive integer $q$.
The collection of Puiseux series over $F$ forms a field 
$F{\{\!\{ X \}\!\}}$ with the obvious definitions of sum and product.  (Note that the product is well-defined thanks to the well-foundedness of the support.)  In case $F$ is an ordered field we can lift the order on $F$ to 
$F{\{\!\{ X \}\!\}}$ by declaring that a non-zero series $f$ as in~\eqref{eq:puiseux} is positive just in case $c_{\vec{\alpha}_0} > 0$, where
$\vec{\alpha}_0$ is the least element of the support of $f$
(w.r.t.~the lexicographic order on $\Q^k$).
We then declare $f < g$ just in case $g - f$ is positive.
If $F$ is a real closed field then 
$F{\{\!\{ X \}\!\}}$ is a real closed field under the above order~\cite[Theorem 2.91]{basu2006algorithms}.\footnote{This result is usually stated in the case of univariate Puiseux series, but the multivariate version follows by induction, since $ F{\{\!\{X\}\!\}}  :=
F{\{\!\{ x_1 \}\!\}}
{\{\!\{ x_2 \}\!\}}
\cdots 
{\{\!\{ x_k \}\!\}}$.}

\subsection{The Strong Solution is Algebraic}
For the rest of this section let $X=(x_1,\ldots,x_k)$ be tuple of commuting indeterminates.
Consider a
proper polynomial system $\mathcal S$ over 
a set of variables $Y=(y_1,\ldots,y_\ell)$, given by equations
\begin{gather}
y_i=P_i \qquad (i=1,\ldots,\ell) 
  \label{eq:sys}
\end{gather}
where $P_i \in \Z[X][Y]$ for all $i\in\{1,\ldots,\ell\}$.
  A solution of $\mathcal S$ in Puiseux series is a tuple
  $(A_1,\ldots,A_\ell) \in \mathbb{R}{\{\!\{X\}\!\}}^\ell$ such
  that evaluating each polynomial $P_i$ at $(A_1,\ldots,A_\ell)$ yields
  an identity $A_i=P_i(A_1,\ldots,A_\ell)$ of Puiseux series. We say that such a solution
  is \emph{non-negative} if $A_i \geq 0$ for all $i=1,\ldots,\ell$.
  In other words, the first coefficient of each of the $A_i$'s is strictly positive, if it exists.

\begin{proposition}
  The strong solution of a proper polynomial system that is defined over $\N$ is the
  least non-negative solution among Puiseux series.
  \label{prop:least}
\end{proposition}
\begin{proof}
  Consider a proper polynomial system $\mathcal{S}$, as
  shown in~\eqref{eq:sys}.  Assume that $\mathcal S$ is defined over $\N$.  Let
  $\boldsymbol{B} \in \mathbb{R}{\{\!\{X\}\!\}}^\ell$ be a non-negative solution of $\mathcal{S}$.  Recall the approximating sequence
  $\boldsymbol{A}^{(0)}, \boldsymbol{A}^{(1)}, \ldots$ of the strong
  solution $\boldsymbol A$ of $\mathcal S$,
  which is defined inductively by $\boldsymbol{A}^{(0)} := \boldsymbol{0}$
  and, for all $n\in\mathbb{N}$,
  \[ \boldsymbol{A}^{(n+1)} :=
  (P_1(\boldsymbol{A}^{(n)}),\ldots,
  P_\ell(\boldsymbol{A}^{(n)})).\]
  Since the integer coefficients in each
  polynomial $P_i$ in $\mathcal{S}$ are non-negative, $P_i$ defines a
  monotone function from $\mathbb{R}{\{\!\{X\}\!\}}^\ell$ to
  $\mathbb{R}{\{\!\{X\}\!\}}$.  
    Now we have $\boldsymbol 0
  \leq \boldsymbol B$ by assumption.  Moreover if we inductively
  assume that $\boldsymbol{A}^{(j)} \leq \boldsymbol{B}$ then,
  since each polynomial $P_i$ has coefficients in $\N[X] \subseteq \Z[X]$, we have
  \begin{eqnarray*}
    \boldsymbol{A}^{(j+1)} &= &(P_1(\boldsymbol{A}^{(j)} ),\ldots,P_\ell(\boldsymbol{A}^{(j)} ))\\
                                          &\leq &
                                                  (P_1(\boldsymbol
                                                  B),\ldots,P_\ell(\boldsymbol
                                                  B))\\
                           &=& \boldsymbol{B} \, .
                             \end{eqnarray*}
We conclude that $\boldsymbol{A}^{(j)} \leq \boldsymbol B$ for all $j$,
and hence  $\boldsymbol A \leq \boldsymbol B$.
\end{proof}

\begin{theorem}
There is an absolute constant $c$ with the following property.
  Given a strong solution $A \in \Z[[X]]$ of a proper
  polynomial system on $\ell$ variables and involving polynomials of total degree at most $d$, there is a polynomial
  $P \in \Z[X][y]$ of total degree at most $d^{c\ell^2}$ 
  such that $P(A)=0$.
  \label{thm:degree}
  \end{theorem}
\begin{proof}
Let us prove the theorem first in the special case that is defined over $\N$. 
Assume that 
$\mathcal S$ is as displayed in Equation~\eqref{eq:sys}.  Consider the following first-order formula
$\Phi(U,v)$ in free variables $U=(u_1,\ldots,u_k)$ and $v$, involving also bound variables $Y=(y_1,\ldots,y_\ell)$ and
$Z=(z_1,\ldots,z_\ell)$:
\[
\begin{array}{rl}
\Phi(U,v):=& \exists Y \, \forall Z \\
& \Big( Y\geq 0 \,  \wedge    Y =
(P_1(U,Y),\ldots,P_\ell(U,Y))\\
 & \wedge  (Z\geq 0 \, \wedge 
\, Z=\big( P_1(U,Z),\ldots,P_\ell(U,Z))\\
& \rightarrow Y \leq Z \big)\,
 \wedge\, v = y_1 \Big)    
\end{array}
\]
Intuitively this formula expresses that $Y$ is the least non-negative solution of $\mathcal S$.

Consider a variable assignment
$\alpha : U \cup \{v\} \rightarrow \mathbb{R}{\{\!\{X\}\!\}}$ satisfying $\alpha(u_i)=x_i$ for $i=1,\ldots,k$.
By Proposition~\ref{prop:least} the unique value of $\alpha(v)$ such
that $\Phi(U,v)$ is satisfied by $\alpha$ is the strong solution of
$\mathcal S$.

By Theorem~\ref{thm:elim} there is a quantifier-free formula
$\Phi'(U,v)$ that has the same set of satisfying assignments as $\Phi$ over
any real closed field, and in particular over 
$\mathbb{R}{\{\!\{X\}\!\}}$,
and that is moreover built over a
family $\mathcal{P} \subseteq \mathbb{Z}[U][v]$ of polynomials of
total degree at most $d^{c\ell^2}$.  Consider again the assignment
$\alpha$ introduced above.  By the uniqueness of the strong solution there must be an
inequality $P\geq 0$ in $\Phi'$
such that the equality $P=0$ holds under assignment
$\alpha$.
(Indeed if none of the inequalities in $\Phi'$ were tight under the assignment $\alpha$ then a suitably small perturbation of the value $\alpha(v)$ would yield second satisfying assignment of $\Psi'$.)
In other words, there is a polynomial $P \in \mathcal{P}$
such that the equation $P(X,A)=0$ holds in
$\mathbb{R}\{\!\{ X \}\!\}$.

It remains to handle the case that $A$ is a solution of a general polynomial system $\mathcal{S}$ with $\ell$ variables and involving polynomials of total degree at most $d$.  Here we can write $A$ as the difference $A=A^{(1)}-A^{(2)}$ of two
series that are components of the solution of a proper system that is
defined over $\N$ and involves
$2\ell$ variables and polynomials of degree at most $d$~\cite[Section IV.2, Theorem 2.4]{SalomaaSoittola}.  The reasoning above shows that
the two series $A^{(1)}$ and $A^{(2)}$ are definable over $\mathbb{R}{\{\!\{X\}\!\}}$ by 
$\exists^* \forall^*$ 
formulas and hence the series $A$ is also definable by such a formula.  The rest of the proof now follows as in the case of a single system that is defined over $\N$.
\end{proof}
\subsection{Bounds on the Order and Degree}
The following proposition relates
the order and degree (when finite)
of an algebraic power series in $\Z[[X]]$ to the degree of its annihilating polynomial.
Given $f=\sum_{v\in\N^k} \alpha_v X^v \in \Z[[X]]$ and 
$D\in\N$, define $$\mathrm{tail}_D(f):=
\displaystyle\sum_{{v\in\N^k,\,|v|>D}} \alpha_v X^v$$
to be the series obtained by deleting all 
monomials of total degree at most $D$.

\begin{proposition}
Let $f=\sum_{v\in\N^k} \alpha_v X^v \in \Z[[X]]$ be a series with annihilating polynomial
$P(y)=\sum_{n=0}^N c_n y^n$, where $c_0,\ldots,c_N\in \Z[X]$ all have total degree at most $D$.  Then the following hold:
\begin{enumerate}
\item if $f\neq 0$ then $\mathrm{ord}(f) \leq D$,
\item if $\mathrm{deg}(f)<\infty$ ($f$ is a polynomial) then $\mathrm{deg}(f) \leq D$,
\item if $\mathrm{deg}(f)=\infty$ then 
$\mathrm{ord}(\mathrm{tail}_D(f)) \leq DN+D$.
\end{enumerate}
\label{prop:annihilate}
\end{proposition}
\begin{proof}
For Item 1 we reason as follows.  
Polynomial $P$ is non-zero by definition of an annihilating polynomial.  
Since $f$ is assumed to be non-zero we can, 
by dividing $P$ by a suitable power of $y$, assume without loss of generality that the constant term $c_0$ is non-zero.
From $P(f)=0$ we have
$c_0 = -\sum_{n=1}^N c_nf^n$ and thus
\begin{eqnarray*}
    D&\geq & \mathrm{ord}(c_0)\\
    &=& \mathrm{ord}\left( \sum_{n=1}^N c_nf^n \right)\\
    &\geq & \min \left\{ \mathrm{ord}(c_nf^n) : n=1,\ldots,N\right\}\\
    &= & \min 
    \left\{\mathrm{ord}(c_n) + n \cdot \mathrm{ord}(f) : n=1,\ldots,N\right\}\\
    &\geq & \mathrm{ord}(f) \, .
\end{eqnarray*}

For Item 2, suppose that $\mathrm{deg}(f)<\infty$,
i.e., $f$ is a polynomial.
Since $P(f)=0$ we have 
$c_Nf^N = - \sum_{n=0}^{N-1} c_nf^n$ and thus
\begin{eqnarray*}
    N\, \mathrm{deg}(f) &\leq & \mathrm{deg}(c_Nf^N) \\
    &=& \mathrm{deg}\left( \sum_{n=0}^{N-1} c_nf^n \right)\\
    &\leq & \max \left\{ \mathrm{deg}(c_nf^n) : n=0,\ldots,N-1\right\}\\
    &\leq & (N-1)\mathrm{deg}(f)+D \, .
\end{eqnarray*}
Hence $\mathrm{deg}(f) \leq D$.

For Item 3, suppose that $\mathrm{deg}(f)=\infty$.
Then $g:=\mathrm{tail}_D(f)$ is a non-zero series with annihilating polynomial
$Q(y):=P(y+f-g)=\sum_{n=0}^N \widehat{c}_n y^n$.
Since $\mathrm{deg}(f-g)\leq D$ we have that $\mathrm{deg}(\widehat{c}_n) \leq ND+D$ for $n=0,\ldots,N$.  Applying Item 1 we conclude that $\mathrm{ord}(g)\leq ND+D$.
\end{proof}

Combining~\cref{prop:annihilate} and~\cref{thm:degree}
we obtain:
\begin{corollary}
\label{cor:orderbound}
Let $A$ be the strong solution of a proper polynomial  system with $\ell$ variables and polynomials of degree at most $d$.
Then for some absolute constant $c$ and $D:=d^{c\ell^2}$ we have:
\begin{enumerate}
\item if $A\neq 0$ then $\mathrm{ord}(A) \leq D$,
\item if $\mathrm{deg}(A)<\infty$, then $\mathrm{deg}(A) \leq D$,
\item if $\mathrm{deg}(A)=\infty$, then
$\mathrm{ord}(\mathrm{tail}_D(A)) \leq D^2+D$.
\end{enumerate}
\end{corollary}

  \begin{proof}
    By~\cref{thm:degree} there is 
    a polynomial $P\in\Z[X][y]$ of total degree at most
$D$ such that $P(A)=0$.  Here $c$ is the absolute constant mentioned in~\cref{thm:degree}.  Items 1--3 of the present result follow immediately from Items 1--3 of~\cref{prop:annihilate}.
      \end{proof}

We will use Item 1 of~\cref{cor:orderbound} to decide {\EqAlg}, and Items 2 and 3 for {\FinAlg}.

\subsection{Putting Things Together}
The first main result of the section is as follows:
\thmAlgIT*

\begin{proof}
Let~$\mathcal{S}$  be a proper polynomial system  of size~$s$, over 
commuting indeterminates~$X=(x_1,\ldots,x_k)$, that has $\ell$ variables and involves polynomials of total degree at most $d$.
Let $A$ be the the formal power series computed by  $\mathcal S$.

Write $D:=d^{c\ell^2}$ for $c$ the absolute constant mentioned in~\cref{cor:orderbound}.
%
By Equation~\eqref{eq:convergence2},
there is a polynomial $\widetilde{a} \in \Z[X]$ such that $\ord(A-\widetilde{a}) \geq D+1$; and by  \Cref{prop:poly_circuit}
there is an algorithms that constructs a circuit $\mathcal C$ representing $\widetilde{a}$ in time $\poly(s)$.

It follows that 
$\ord(\widetilde{a}) \geq D+1$ if $A$ is identically zero 
and, by~\cref{cor:orderbound}(1), $\ord(\widetilde{a}) \leq D$
if $A$ is not identically zero.

Define $D':=2^{s'}$ where $s'$ is the size of the circuit~$\mathcal C$. Note that $D'$ is an upper bound of the degree of the polynomial represented by~$\mathcal C$.
One can construct a circuit~$\mathcal{C}'$ from~$\mathcal C$, in time $\poly(s)$,
for the polynomial 
\[ f:=(x_1 \cdots x_k)^{D'} \cdot \widetilde{a}\left(\frac{1}{x_1},\ldots,\frac{1}{x_k}\right) \, .\]
If $A$ is identically zero then the total degree of~$f$ is at most~$kD'-D-1$, whereas  if $A$ is non-zero then the degree is at least~$kD'-D$. 
\end{proof}

We now come to the second main result of the section, concerning the problem {\FinAlg}.  Before proceeding with the proof we observe that, unlike the approximating sequence of the strong solution, the sequence of iterates defined in Hensel's Lemma may fail to stabilise in finitely many steps even when the target sequence has finite support.

\begin{example}
Consider again the proper system $y=x+x^2-2xy+y^2$
discussed in Example~\ref{ex:iteration-quasiregular}. 
Although the unique quasiregular solution~$x$ has finite support, Hensel's iteration 
does not stabilise in finitely many steps, as witnessed by the 
 closed-form formula  for~$a_n$.
 
 %
\end{example}

\thmFinAlg*
\begin{proof}
Let~$\mathcal{S}$  be a proper polynomial system of size~$s$, over 
commuting indeterminates~$X=(x_1,\ldots,x_k)$, that has $\ell$ variables and involves polynomials of total degree at most $d$.
Let $A$ be the the formal power series computed by  $\mathcal S$.

Write $D:=d^{c\ell^2}$ for $c$ the absolute constant mentioned in~\cref{cor:orderbound}. 
By Equation~\eqref{eq:convergence2},
there is a polynomial $\widetilde{a} \in \Z[X]$ such that $\ord(A-\widetilde{a}) \geq D^2+D$; and by  \Cref{prop:poly_circuit}
there is an algorithm that constructs a circuit $\mathcal C$ representing $\widetilde{a}$ in time $\poly(s)$.


Suppose that $A$ has finite support.  Then $\deg(A)\leq D$ by~\cref{cor:orderbound}(2) and hence $\widetilde{a}$ contains no monomial of total degree in the interval $[D+1,D^2+D]$. 
On the other hand, if $A$ has infinite support then by~\cref{cor:orderbound}(3)
$\mathrm{tail}_D(A)$ and hence also $\widetilde{a}$ contains some monomial of total degree $[D+1,D^2+D]$.  
We conclude that $A$ has infinite support if and only if
the polynomial $\widetilde{a}$ contains a monomial of total degree lying in the interval $[D+1,D^2+D]$.
This monomial can be guessed and then checked for non-zeroness using an oracle for {\CoeffSLP}.
\end{proof}

\section{Applications to Context-Free Grammars} 
\label{sec:cfg}
The \emph{Multiplicity Equivalence Problem} for context-free grammars
asks, given two grammars $G_1$ and $G_2$ and respective  
non-terminals $N_1, N_2$ thereof, whether each word has the same number of derivations starting in $N_1$ as in $N_2$ (see below for formal definitions).  By taking the disjoint union of the two grammars, one may assume without loss of generality that $N_1$ and $N_2$ are non-terminals of the same grammar.

%
%
Decidability of multiplicity equivalence for grammars
is a long-standing open problem in the theory of formal languages.  
It generalises decidability of language equivalence of unambiguous grammars, itself a recognised open problem, as well as decidability of language equivalence of 
deterministic pushdown automata, established in~\cite{Senizergues:ICALP:1997} (see also \cite{Jancar:LICS:2012}).
%
%
For the special case of linear context-free grammars with a distinguished symbol marking the middle of the word,
multiplicity equivalence reduces to multiplicity equivalence of two-tape finite automata,
which is known to be decidable for any number of tapes \cite{HarjuKarhumaki:TCS:1991}.


\subsection{Context-Free Grammars}
Let $\Sigma=\{\sigma_1,\ldots,\sigma_k\}$ be a finite alphabet.
The \emph{Parikh image} is the function $c:\Sigma^* \to \N^k$
such that for all words~$w$, we have 
$c(w)=(v_1,\ldots,v_k)$ where the~$v_i$ is the number of occurrences of letter~$\sigma_i$ in~$w$.

A \emph{context-free grammar} is a tuple~$G=(\Sigma, V, \Delta)$
where $\Sigma$ is a finite \emph{alphabet}, 
$V$ is a set of \emph{nonterminals}, 
and $\Delta \subseteq V \times (V \cup \Sigma)^+$ is a set of \emph{production rules}.
%
We write the production rules in the form $N \rightarrow \alpha$
where $N \in V$ and $\alpha \in (V \cup \Sigma)^+$.
In this paper, we assume that the grammars are \emph{proper}, that is the right-hand side~$\alpha$ of 
each rule $N \rightarrow \alpha$ is non-empty and does not consist of a single non-terminal.


For the grammar~$G$,
the binary relation $\Rightarrow$ on~$(V \cup \Sigma)^*$,
capturing a \emph{(leftmost) derivation step}, is defined as follows:
if $N \rightarrow \alpha$ is in $\Delta$ then
for all words $\beta \in \Sigma^*$ and $\gamma\in (V\cup\Sigma)^*$
 we have 
 $\beta N\gamma\, \Rightarrow \, \beta \alpha \gamma$.
A sequence $\alpha_0 \Rightarrow \alpha_1 \Rightarrow \cdots \Rightarrow \alpha_k$ is a \emph{derivation}, of $\alpha_k$ from $\alpha_0$. 

The syntactic condition that $G$ is proper implies that for all words~$w \in \Sigma^*$, and all nonterminals~$Y\in V$,
the number of distinct derivations of~$w$ starting from $Y$ is finite.
For instance, for the improper grammar $X \rightarrow a, X \rightarrow X$
the word $a$ has infinitely many derivations from $X$.
We define the following \emph{multiplicity semantics} $\sem N : \Sigma^* \to \N$
of a nonterminal $N$ of a proper grammar $G$:
For every finite word $w\in \Sigma^*$, $\sem N_w$
is the number of distinct derivations of $w$ starting from the nonterminal~$N$.
The language of~$N$, denoted $L(N)$, is the set of words~$w$ such that $\sem N_w \neq 0$.

The \emph{Multiplicity Equivalence Problem} for a given grammar~$G$ and two nonterminals $N_1, N_2$ thereof
asks whether $\sem{N_1} = \sem{N_2}$.

\subsection{Letter-bounded languages}
\label{sec:letter-bounded}
We say that a language $L \subseteq \Sigma^*$ is \emph{letter-bounded} if there is an permutation
$\sigma_1,\ldots,\sigma_k$ of $\Sigma$ such that $L \subseteq \sigma_1^* \cdots \sigma_k^*$.
Deciding whether a given context-free language $L(G)$ is letter-bounded
and moreover finding a witnessing enumeration of letters can be done in polynomial time
by a simple dynamic programming algorithm:
%
%
Then, we compute in polynomial time the set of letters that can appear as the first letter in some word of $L(G)$.
If this set is not a singleton, then $L(G)$ is not letter-bounded.
Otherwise, this set contains a single letter $\sigma_1$.
Compute in polynomial time a grammar $G_1$
recognising the context-free language over $\Sigma_1 = \Sigma \setminus \{\sigma_1 \}$
equal to $\{ w \mid \exists n \in \N: \sigma_1^n w \in L(G), \sigma_1^{n+1} w \not \in L(G)\}$.
The procedure can be applied inductively to $G_1$ over the smaller alphabet $\Sigma_1$.
Since at each step we remove one letter from the alphabet,
after $k$ steps we reach a grammar $G_k$ over the empty alphabet $\Sigma_k = \emptyset$.
Then $L(G)$ is letter-bounded iff $L(G_k) = \emptyset$;
in the positive case, the algorithm has constructed a witnessing enumeration $\sigma_1^* \cdots \sigma_k^*$.
Moreover, each of the $k$ steps is performed in polynomial time,
so we have an overall polynomial time complexity.

%


We show how to decide multiplicity equivalence of non-terminals
generating a letter-bounded language.
Consider a nonterminal $N$.
Recall that $\sem N_w \in \N$, for a word $w \in \Sigma^*$,
is the multiplicity of $w$ as generated by $N$.
We now aggregate the multiplicities of all words with the same Parikh image.
Let $X=\tuple{x_1, \dots, x_k}$ be a tuple of commuting indeterminates, with one variable $x_i$ for each terminal symbol $\sigma_i$. Given a non-terminal $N$,
define its \emph{census} generating function~\cite{Panholzer05} to be the multivariate power series 
\begin{align*}
    f_N := \sum_{\vec v \in \N^k} a_{\vec v}(N) \cdot X^{\vec v} \, ,
\end{align*}
where 
\[a_{\vec v}(N) := \sum_{c(w) = \vec v} \sem N_w.\]
Thus defined, the tuple of formal series $(f_N)_{N \in V}$ satisfies a proper polynomial system
 over indeterminates $X$ that can be obtained from the grammar $G$ in polynomial time~\cite[Theorem 1.5 in Chapter IV]{SalomaaSoittola}.
This system is obtained by a classic syntactic transformation applied to the grammar.
Rather than formally defining it here, we present it with an example.

\begin{example}
    Consider the proper grammar $G$ over the alphabet of terminal symbols $\Sigma = \{a, b, c, d\}$
    with nonterminal symbols $V = \{X, Y\}$ and a production rules
    \begin{align*}
        X \rightarrow ab, \ X\rightarrow a X X b, \ X \rightarrow c Y d, \ Y \rightarrow c d, \ Y \rightarrow c Y Y d.
    \end{align*}
    %
    %
    We obtain the proper system of polynomial equations:
    \begin{align*}
        f_X = x_1 x_2 + x_1 x_2 f_X^2 + x_3 x_4 f_Y \quad f_Y = x_3 x_4 + x_3 x_4 f_Y^2.
    \end{align*}
\end{example}


Multiplicity equivalence of letter-bounded context-free grammars reduces to equivalence of census generating functions.
\begin{lemma}
    \label{lem:reduction from letter-bounded to EqAlg}
    Let $N_1, N_2$ be  non-terminals of a grammar such that $L(N_1), L(N_2) \subseteq \sigma_1^*\cdots \sigma_k^*$.
    Then
    \begin{align*}
        \sem{N_1} = \sem{N_2}
        \quad\text{if and only if}\quad
        f_{N_1} = f_{N_2}.
    \end{align*}
\end{lemma}

\begin{proof}
    This follows at once from the fact that the Parikh image $c$
    restricted to $\sigma_1^* \cdots \sigma_k^*$
    is a bijection onto $\N^k$.
    As a consequence, for every non-terminal~$N$ and every vector $\vec v = \tuple {v_1, \dots, v_k} \in \N^k$ 
    we have
    \begin{align*}
        \sem N_{\sigma_1^{v_1}\cdots\sigma_k^{v_k}} = a_{\vec v}(N). 
    \end{align*}
\end{proof}

    \Cref{lem:reduction from letter-bounded to EqAlg} shows that multiplicity equivalence of letter-bounded context-free grammars
    is a special case of {\EqAlg} for the census generating function.
    By \cref{thm:AlgIT} the latter can be decided in $\coRP^{\PP}$,
    thus proving \cref{cor:multiplicity equivalence-both}(1) from the introduction.

%

\subsection{Bounded context-free languages}
A language $L \subseteq \Sigma^*$ is \emph{bounded} 
if there exist nonempty words $w_1, \dots, w_k \in \Sigma^+$
such that $L \subseteq w_1^* \cdots w_k^*$.
Many algorithmic problems on context-free grammars
are more tractable on bounded languages.
Checking whether a given context-free grammar recognises a bounded language is decidable \cite[Theorem 5.5.2]{Ginsburg:1966},
and can be done in polynomial time \cite[Theorem 19]{Gawrychowski:Krieger:Rampersad:Shallit:IJFCS:2010}.
We note however that there are grammars recognising a bounded language
where the number of witnessing words $k$ is exponential in the size of the grammar.

In this section we give complexity bounds for deciding multiplicity equivalence for arbitrary grammars \emph{restricted} to a bounded language~$L:=w_1^* \cdots w_k^*$ that is explicitly given by the list of words~$w_1,\ldots,w_k$. This problem asks to decide whether $\sem{N_1}_w=\sem{N_2}_w$ for all words~$w\in L$.
We reduce the restricted multiplicity equivalence problem
to the letter-bounded case.
%
%
\begin{restatable}{lemma}{lemmainvhomo}
    \label{lem:invhomo}
    The Multiplicity Equivalence Problem restricted to a bounded language
    reduces in polynomial time to the multiplicity equivalence problem for grammars
    recognising a letter-bounded language.
\end{restatable}
\begin{proof}
    Suppose we wish to check multiplicity equivalence of two non-terminals of a grammar $G$, restricted to a bounded language 
    $L:=w_1^*\cdots w_k^*$.
    Consider a fresh alphabet $\Gamma = \set{a_1, \dots, a_k}$
    and define the homomorphism $h : \Gamma^* \to \Sigma^*$
    by setting $h(a_1) = w_1, \dots, h(a_k) = w_k$.
    Below, we combine classical constructions to 
    transform~$G$ into a new grammar $G'$
    with the property that for each non-terminal~$N$ of $G$ there is a non-terminal $N'$ of $G'$ such that
     \begin{align}
        \label{eq:reduction}
            \sem {N'}_w = \sem N_{h(w)} \, .
    \end{align}    
    %
    This transformation is done in three steps.
    In the first step, as in \cite[Theorem 14.33]{KuichSalomaa}, we convert in polynomial time the non-terminal $N$
    to a pushdown automaton $A$ with the same multiplicity semantics:
    $\sem N = \sem A$ where $\sem A$ is the function that maps each word~$w$ to the number of accepting runs of~$A$ over~$w$.

    In the second step, from $A$ we build a pushdown automaton $B$
    recognising the inverse homomorphic image of the language recognised by $A$:
    $L(B) = h^{-1}(L(A))$.
    This can be achieved by a standard construction on pushdown automata
    \cite[Theorem 7.30]{HopcroftMotwaniUllman:2000},
    which can be performed in polynomial time.
    Inspecting  the construction,  we observe that that multiplicities are preserved, meaning that,
    for every word~$w \in \Gamma^*$, the equality
    $\sem A_{h(w)} = \sem B_w$ holds
    (see also \cite[Claim 11]{Raz:DLT:1993}).
    
    In the third step, as in \cite[Theorem 14.15]{KuichSalomaa}, we convert  in polynomial time the pushdown automaton~$B$
    to a nonterminal $N'$
    having the same multiplicity semantics:
    $\sem B = \sem {N'}$.
    This establishes \eqref{eq:reduction}.

    The non-terminal $N'$ defined above need not recognise a letter-bounded 
    language.
    This can be remedied by taking the product of $N'$
    with a deterministic finite automaton $A$ for the language $a_1^* \cdots a_k^*$,
    which is multiplicity preserving and can be done in polynomial time
    (as described in~\cite[Chapter IV, Theorem~3.5]{SalomaaSoittola}).
    Let $N_1''$ and $N_2''$ be the nonterminals obtained from
    the product of $N_1'$, resp., $N_2'$ with the automaton $A$.
    Not only are $N_1''$ and $N_2''$ letter bounded, but 
    $\sem {N_1}$ and $\sem {N_2}$ coincide on $L$ just in case 
    $\sem {N''_1}=\sem {N''_2}$.
\end{proof}

Combining \cref{cor:multiplicity equivalence-both}(1) and~\Cref{lem:invhomo},
we obtain \cref{cor:multiplicity equivalence-both}(2) from the introduction.


\color{black}

\section{Discussion}
In this paper we have investigated a number of computational problems
concerning the coefficients of algebraic series
defined by systems of polynomial equations
and have related
these to analogous well-known problems
for arithmetic circuits.
In~\Cref{thm:CoeffAlg} we showed that the problems {\CoeffAlg} and
{\CoeffSLP} are polynomial-time interreducible.  A natural question
for future work is whether it is likewise possible to reduce {\EqAlg}
to {\EqSLP}, rendering these two problems equivalent under
polynomial-time reductions. In~\Cref{thm:AlgIT} we gave a reduction of
{\EqAlg} to {\DegSLP}. It is easily seen that {\EqSLP} reduces in
polynomial time to {\DegSLP}, but it is not known whether there is a
polynomial-time reduction in the other direction.  There is moreover
a significant difference in the best known complexity upper bounds for
the two problems: {\EqSLP} (i.e., polynomial identity testing) famously admits
a number of different randomised polynomial-time algorithms, whereas
the best known complexity bound for {\DegSLP} involves a randomised
polynomial-time algorithm with a {\PP} oracle.

%



\color{black}

 \section*{Acknowledgment}

 We would like to thank Wojciech Czerwiński for pointing out reference \cite{Gawrychowski:Krieger:Rampersad:Shallit:IJFCS:2010} to us.
 We would like to thank Keith Conrad for his invaluable collection of expository papers,
 in particular his note on the multivariate version of Hensel's lemma.

\bibliographystyle{IEEEtran}
\bibliography{literature}

\newpage 

\onecolumn
\appendix
\subsection{Extended Preliminaries}
\label{sec:extendedPre}

%
%
%
%
%
%
%
%

\propgpcircuit*
\begin{proof} 
Define $S_{m} := \sum_{i=0}^{m} x^i$. 
Note that
$$ \begin{bmatrix}
S_{m+1} \\ 1
\end{bmatrix}
=  \begin{bmatrix}
x & 1 \\ 0 & 1 \\
\end{bmatrix} \cdot  \begin{bmatrix}
S_{m} \\ 1
\end{bmatrix}
$$
and thus
$$\begin{bmatrix}
S_m \\ 1 
\end{bmatrix} = \begin{bmatrix}
x & 1 \\ 0 & 1 \\
\end{bmatrix}^{m} \begin{bmatrix}
1 \\ 1
\end{bmatrix} $$
for all $m$.  
Since exponentiation of a matrix to the power~$m$ can be implemented via $O(\log{m})$ steps of repeated squaring, the statement follows. 

\end{proof}

\claimvalpoly*
\begin{proof}
    Let $p(x) = c_0 x^0 + \cdots + c_n x^n$.
    We then have $p(a) - p(b) = (a - b) \cdot q(a, b)$ for some polynomial $q(x, y) \in R[x, y]$.
    Thus $v(p(a) - p(b)) = v(a - b) + v(q(a, b)) \geq v(a - b)$,
    where the last inequality follows from the fact that the valuation is nonnegative.
\end{proof}





\subsection{Missing proofs in \texorpdfstring{\Cref{sec:coeffAlg}}{Section~III}}
\label{sec:app-proofs-sec-iii}

\propmultivarhenseljacobian*
\begin{proof}
Recall that the derivative matrix of~$\vec{f}$ is 
\begin{align}
      D\mathbf{f} &=  \begin{bmatrix} \frac{\partial f_1}{\partial y_1} \, \cdots \, \frac{\partial f_\ell}{\partial y_1} \\ \vdots \, \ddots \, \vdots \\ \frac{\partial f_1}{\partial y_\ell} \, \cdots \, \frac{\partial f_\ell}{\partial y_\ell} \end{bmatrix} 
      \end{align}
    
     Recall also that  $\mathcal{S}$ is assumed to be a proper equation system. This requires that for each polynomial $P_i$ and all 
     monomial~$Y^{\vec{v}}$ of total degree at most~1 that appears in $P_i$, 
     the coefficient~$a_{\vec{v}}\in \Z[X]$  lies in $\mathfrak{m}$.  Consequently, we can write $f_i$ as 
    \[
    f_i = y_i - h_{i,0}-\sum_{j=1}^\ell h_{i,j}y_j- \sum_{|\vec{v}|>1} g_{i,\vec{v}}\, Y^{\vec{v}} 
    \]
    for some  $ g_{i,\vec{v}} \in R_0$ and  $h_{i,j} \in \mathfrak{m}$  for all~$j\in \{0,\ldots,\ell\}$. 


    Let $i,j\in \{1,\ldots,\ell\}$ be such that $i\neq j$.
  We can write $\frac{\partial f_i}{\partial y_i}\in R_0[Y]$ as 
    \[\frac{\partial f_i}{\partial y_i} = 1 -  h_{i,i} -  \sum_{\substack{|\vec{v}|>1 \\ v_i \geq 1 }}\, v_i \, g_{i,\vec{v}} Y^{\vec{v}-e_i} \]
    where  $e_i$ is the $i$-th vector in the standard basis (with a $1$ in the $i$-th coordinate and $0$ elsewhere).
    Analogously, we can write~$\frac{\partial f_i}{\partial y_j}\in R_0[Y]$ as  
  \[\frac{\partial f_i}{\partial y_j} = -  h_{i,j} -  \sum_{\substack{|\vec{v}|>1 \\ v_j \geq 1 }}\, v_i \, g_{i,\vec{v}} Y^{\vec{v}-e_j} \] 
Consequently, for all $\vec{a}\in \mathfrak{m}^\ell$, we have     
 \[\frac{\partial f_i}{\partial y_i} (\vec{a}) \in 1 +\mathfrak{m}\]

   and  
    \[\frac{\partial f_i}{\partial y_j} (\vec{a}) \in \mathfrak{m} \, .\]
Thus $D\vec{f}$ is entry-wise congruent to the identity matrix modulo $\mathfrak{m}$.  Since the determinant of a matrix is a polynomial in its entries, it follows that $J_{\vec{f}}(\vec{a})$
is congruent modulo $\mathfrak{m}$ to the determinant of the identity matrix, which proves the claim.




\end{proof}

\propmultivarhenselnorm*
\begin{proof}

    Recall  that  $\mathcal{S}$ is assumed to be a proper equation system. This requires that for each polynomial $P_i$ and all 
     monomials~$Y^{\vec{v}}$ of total degree at most~1 that appear in $P_i$, 
     the coefficient~$a_{\vec{v}}\in \Z[X]$  lies in $\mathfrak{m}$.  Suppose that $\vec{a}\in \mathfrak{m}$. Then
    \[
    f_i(\vec{a}) = y_i - P_i(\vec{a}) 
    \]
    But $P_i(\vec{a}) \in \mathfrak{m}$ since  the coefficient of~$Y^{\vec{0}}$  lies in~$\mathfrak{m}$. 
\end{proof}

\propnewtonratfunction*
\noindent
We recall here Equation~(\ref{eq:a n i}) mentioned in the statement of the claim:
\begin{align}
    \tag{\ref{eq:a n i}}
    {\bf a}_{n,i} =\frac{g_{n,i}}{1- h_{n,i}}, \quad \text{with } \ord(h_{n,i})\geq 1.
\end{align}
\begin{proof}[Proof of~\Cref{prop:newton_rat_function}]
Recall that ${\bf a}_0 = \bf 0$ and that we have the recursive formula~\eqref{eq:update}. 
We define the circuits $C_{n,i}$ and $D_{n,i}$ 
by induction on $n$. 
For each~$i \in \{1,\ldots,\ell\}$ we define 
$C_{0,i}$ to be the constant~$0$ and $D_{0,i}$ to be the constant~$1$. 
For the inductive step we show how to construct for each $i \in \{1,\ldots,\ell\}$
the circuits $C_{n+1,i}$ and $D_{n+1,i}$ from the collection of circuits 
$\{C_{n,j},D_{n,j} : 1\leq j \leq \ell\}$, in time 
${\bf poly}(s)$.  For this we will use Equation~\eqref{eq:update}.  

\Cref{prop:det-circuit} ensures we can construct $\poly(s)$-size circuits for $J_{\bf f}({\bf a}_{n})$ as well as every entry of adjugate $\mbox{Adj}(D{\bf f}({\bf a}_{n}))$,
whose entries are just cofactors of~$D{\bf f}({\bf a}_{n})$.
Composing the circuits so-obtained with the circuits  
~$\{C_{n,j},D_{n,j} : 1\leq j \leq \ell\}$
 that represent the respective numerators and  denominators of the entries of ${\bf a}_n$ we obtain the desired circuits $C_{n+1,i}$ and $D_{n+1,i}$. This can be easily done by 
 using the rules $\frac{A}{B} + \frac{C}{D} = \frac{AD+BC}{BD}$ and $\frac{A}{B}\cdot \frac{C}{D} = \frac{AC}{BD}$.

It remains to argue that $\ord(h_{n,i})\geq 1$  for the polynomial~$h_{n,i}\in R_0$ represented by $D_{n,i}$, where
 $\vec{a}_{n,i}=\frac{g_{n,i}}{1-h_{n,i}}$.  
Since $J_{\bf f}(.)$ is a polynomial map, it maps $(\Q(X)\cap R)^\ell$ to the ring~$\Q(X)\cap R$.
In particular, we can write $J_{\bf f}(\vec{a}_{n-1})$  as $\frac{s}{1-t}$ for some $s,t\in R_0$. 
Moreover, 
by~\Cref{prop:multivar_hensel_jacobian} we know that~$J_{\bf f}(\vec{a}_{n-1})\in 1+\mathfrak{m}$,
hence the constant term of $s$ must be $1$. 
Since $h_{n,i}$ is obtained through  multiplication of $s$  and 
some of the polynomials $1-h_{n-1,j}$ with $j\in \{1,\ldots, \ell\}$ and $\ord(h_{n-1,j})\geq 1$, 
 the denominator of~$\vec{a}_{n,i}$ will have constant term~$1$. The claim follows. 
\end{proof}

\proppolycircuit*
\begin{proof}
By definition in Equation~\eqref{eq-approx-def}, we have $\widetilde{\bf a}_{n,1}=g_{n,1} \sum_{j=0}^{2^n-1} h_{n,1}^j$.  Note that a naive circuit for the expression on the right-hand side would have size $\poly(s)2^n$. However one can build a circuit of size $\poly(s,n)$ by plugging in $x = h_{n,1}$ and $m = 2^n$ in the circuit from Proposition~\ref{prop:gpcircuit}.
 \end{proof}
 
\sharpPhardCoeffAlg*
\begin{proof}
We present a straightforward reduction from  {\CoeffSLP} to  
{\CoeffAlg}. 
Fix an instance of {\CoeffSLP}, comprising a polynomial~$f$ represented by an arithmetic circuit~$C$ over   variables~$X$, a monomial  $X^{\bold{v}}$ and a prime~$p$.

We say that $C$ is balanced if all paths from the input gates to the designated output gate  have equal length.  
It is folklore that an algebraic circuit of size $s$ can be transformed into a binary and balanced circuit of size $O(s)$ computing the same polynomial in time $O(s)$.
Furthermore, the equivalent circuit can be constructed to have alternating levels of multiplication and addition/subtraction gates.
We thus assume without loss of generality that the circuit $C$ is binary, balanced and alternating.

Fix $\tilde{x}\in X$. We first replace each input gate $m\in \{0,1\}\cup X$ in~$C$  with~$m \tilde{x}$. Next, we 
consider the corresponding SLP of~$C$. That is  a collection of  sequential instructions in the form \begin{align*}
	Y_i & :=  Y_j \odot Y_{k}
\end{align*}
with $\odot \in \{\times,+,-\}$  where the $Y_i$ is an $\odot$-gate,  and the $Y_j,Y_{k}$ are the  two inputs of~$Y_i$ in~$C$.

We modify these equations  to obtain a proper equation system~$\mathcal{S}$. The idea is to merge three equations arising from a multiplication followed by  additions/subtractions into a single equation. The construction merges equations such as 
    \begin{align*}
	Y_1 & :=  Y_2 \times Y_3\\
    Y_2 &:= Y_4 +  Y_5 \\
    Y_3 &:= Y_6 + Y_7  
\end{align*}
with  all the $Y_i$  variables in the SLP,  to the single equation 
\[Y_1  :=  Y_4 Y_6 + Y_4 Y_7 + Y_5 Y_6 + Y_5 Y_7.\] 
A similar transformation is applied to a multiplication gate whose inputs are two subtraction gates, or an addition and a subtraction gate. 

We argue  the power series computed by~$\mathcal{S}$, that is clearly proper, is $f$ multiplied with $\tilde{x}^{\alpha}$ for some $\alpha \geq 1$.
The proof follows from the assumption that $C$ is balanced. The proof is by an induction
showing that all gates with the same distance to the input gates are multiplied with $\tilde{x}^{\alpha}$ for some~$\alpha$. For the base of induction, observe that
 all input gates $m$ are replaced with $m \tilde{x}$. For the induction step, given that two inputs $Y_j$ and  $Y_{m}$ of a $\odot$-gate~$Y_i$ are both multiplied with $x^{\alpha}$, then $Y_i$ is multiplied with 
\begin{itemize}
    \item $\tilde{x}^{\alpha}$ if $\odot\in \{+,-\}$,
    \item  $\tilde{x}^{2\alpha}$ otherwise (i.e., $\odot=\times$). 
\end{itemize}
We note that one can compute in ${\sf NC}$ an $\alpha$  such that the power series computed by~$\mathcal{S}$ is precisely  $\tilde{x}^{\alpha}f $. 
The statement of the claim follows.
\end{proof}

\subsection{Proof of Equation~\texorpdfstring{\eqref{eq:example-closed-form}}{\ref{eq:example-closed-form}}}
\label{sec:app-example}

We will prove by induction that, for all $n \geq 1$, the closed-form formula \eqref{eq:example-closed-form} for $a_n$ holds. 
It is straightforward to verify that $a_1 = x - \frac{x^2}{(x+1)^2 -x^2}$. For the inductive step, assume that~\eqref{eq:example-closed-form} holds for some $n \in \N$. Following \eqref{eq:example-rec} we can write
\begin{equation}
\setlength{\jot}{12pt}
\begin{aligned}
    a_{n+1} &= \frac{-(x - \frac{x^{2^n}}{(x+1)^{2^n} - x^{2^n}})^2 + x^2 + x}{1 + 2x - 2(x - \frac{x^{2^n}}{(x+1)^{2^n} - x^{2^n}})}
    \\
    &= \frac{2x\frac{x^{2^n}}{(x+1)^{2^n} - x^{2^n}} - \frac{(x^{2^n})^2}{((x+1)^{2^n} - x^{2^n})^2} + x}{1 + \frac{x^{2^n}}{(x+1)^{2^n} - x^{2^n}}}
    \\
    &= \frac{\displaystyle
    2x \cdot x^{2^n}(x+1)^{2^n} - 2x \cdot x^{2^{(n+1)}} - x^{2^{(n+1)}} + x\cdot((x+1)^{2^n} - x^{2^n})^2}
    {\displaystyle((x+1)^{2^n} - x^{2^n} + 2x^{2^n})((x+1)^{2^n} - x^{2^n})}
    \\
    &= \frac{x\cdot(x+1)^{(n+1)}-x \cdot x^{2^{(n+1)}} - x^{2^{(n+1)}}}{((x+1)^{2^n} + x^{2^n})((x+1)^{2^n} - x^{2^n})}
    \\
    &= \frac{x\cdot(x+1)^{(n+1)}-x \cdot x^{2^{(n+1)}} - x^{2^{(n+1)}}}
    {(x+1)^{2^{(n+1)}} - x^{2^{(n+1)}}}
\end{aligned}
\end{equation}
Hence $a_{n+1} = x - \frac{x^{2^{(n+1)}}}{(x+1)^{2^{(n+1)}} - x^{2^{(n+1)}}}$
and our induction is completed.

\end{document}